\documentclass[sigconf]{acmart}
\usepackage{flushend}   
\usepackage{appendix} 
\usepackage{amsthm} 
\newtheorem{remark}{Remark}
\usepackage{bm} 
\usepackage{float}  
\usepackage{enumitem} 
\usepackage{multirow}

\newtheorem{Corollary}{Corollary}
\newtheorem{proposition}{proposition}

\AtBeginDocument{%
  \providecommand\BibTeX{{%
    \normalfont B\kern0.5em{\scshape i\kern0.25em b}\kern0.8em\TeX}}}

\setcopyright{acmlicensed}
\copyrightyear{2024}
\acmYear{2024}
\acmDOI{XXXXXXX.XXXXXXX}

\acmConference[Conference acronym 'XX]{Make sure to enter the correct
  conference title from your rights confirmation email}{June 0305,
  2018}{Woodstock, NY}
\acmISBN{97814503XXXXX/18/06}

\graphicspath{{./images/}}

\begin{document}

\title{PSG: Pair-Space Generation for Efficient Generative Reranking}

\author{Chao Feng}

\affiliation{%
  \institution{Kuaishou Tech}
  \city{Beijing}
  \country{China}
}
\email{fengchao08@kuaishou.com}

\author{Li Ma}
\affiliation{%
  \institution{Kuaishou Tech}
  \city{Beijing}
  \country{China}}
\email{maili06@kuaishou.com}

\author{Xiancheng Gao}
\affiliation{%
 \institution{Kuaishou Tech}
 \city{Beijing}
 \country{China}}
\email{gaoxiancheng@kuaishou.com}

\author{Chenghao Zhang}
\affiliation{%
  \institution{Kuaishou Tech}
  \city{Beijing}
  \country{China}
}
\email{zhangchenghao03@kuaishou.com}

\author{Yuanhao Pu}
\affiliation{%
  \institution{Kuaishou Tech}
  \city{Beijing}
  \country{China}}
\email{puyuanhao@kuaishou.com}
\author{Xiang Li}
\affiliation{%
  \institution{Kuaishou Tech}
  \city{Beijing}
  \country{China}}
\email{lixiang44@kuaishou.com}





\renewcommand{\shortauthors}{Chao F, et al.}

\begin{abstract}
Modern recommender systems adopt Generator-Evaluator (G-E) for list-wise reranking: a generator produces sequences from candidates and an evaluator scores them at sequence-level to filter out the optimal one for exposure. Auto-Regressive(AR), working as the backbone for generative recommendation, suffers two limitations. First, its complexity grows linearly with list length, forcing the system to generate fewer lists under rigorous latency constraints and thus limiting exploration. Second, teacher-forcing creates a train–test mismatch; cumulative errors worsen with length and degrade quality.

To address these problems, we propose Pair-Space Generation (PSG), a reformulation that elevates the generation atom from individual items to ordered item pairs. Given $n$ candidate items, PSG operates over pair vocabulary of size $n(n-1)$ per request, generates only $L/2$ tokens. Pair token representations are produced on-the-fly by a pretrained pair-token representation module optimized over large scale exposure logs, eliminating the data sparsity that would otherwise plague a quadratic sized vocabulary.
We establish three theoretical guarantees: (i) PSG is bijective with item-space generation and induces an equivalent family of sequence distributions, thus incurring no loss of expressiveness; (ii) generation in pair-token space achieves approximately a $2\times$ to $4\times$ speedup theoretically under moderate settings and $1.83\times$ in the real industrial environmental settings; and (iii) under outcome-only rewards, the worst-case suboptimality of PSG is bounded by $O((L/2)^2 \bar{\epsilon})$, representing a nearly $4\times$ improvement over item-space generation. Beyond benchmark-based validation, PSG has also been deployed on Kuaishou, delivering a 0.178\% lift in per-user stay time on the platform, which serves over 400 million daily active users.
\end{abstract}

\ccsdesc[500]{Information systems~Recommender systems}
\ccsdesc[500]{Information systems~Learning to rank}

\keywords{Recommendation, Generative re-ranking, Generator-Evaluator, Pair-Token, Efficiency, Accumulative Error reduction}

\received{20 February 2024}  
\received[revised]{12 March 2024} 
\received[accepted]{5 June 2024}  

\maketitle

\section{Introduction}

Reranking is the final stage of industrial recommender pipelines. Given a candidate set of $n$ items returned by upstream retrieval and ranking, reranking produces an ordered list of length $L$ that maximizes a composite list-level utility (e.g., engagement, dwell time, monetization). In contrast to pointwise scoring, the defining characteristic of reranking is the dominance of item–item interactions: complementary items lift each other, redundant items diminish diversity, and position cascades amplify exposure asymmetries.

The solution space of reranking consists of all ordered permutations of $L$ items selected from $n$ candidates, with size $P(n,L)=\frac{n!}{(n-L)!}$, which grows exponentially with both $L$ and $n$. Unlike upstream ranking, reranking utility must account for combinatorial effects (e.g., diversity) among items. To capture such inter-item influences within a list, early works address this via context-aware pointwise scoring (e.g., DLCM\cite{DBLP:conf/sigir/AiBGC18} and PRM\cite{DBLP:conf/recsys/PeiZZSLSWJGOP19}, Set-Rank\cite{DBLP:conf/sigir/PangXALCW20}, PEAR\cite{DBLP:conf/www/LiZLSCZTXH22}, MIR\cite{DBLP:conf/sigir/XiLZZDTZZY22}), yet such methods only partially explore the combinatorial solution space. In contrast to pointwise scoring with contextual refinement, the G-E\cite{DBLP:conf/sigir/YangQ0YWLHLG25} framework decouples generation from evaluation, thereby naturally expanding the search space. The generator produces candidate sequences from the exponential solution space via sequential decoding or sampling strategies; the evaluator scores each sequence on diversity, engagement, and revenue, and selects the top-1 for final exposure. Evaluator feedback, in turn, steers the generator toward regions of higher utility.


In this paper, we primarily focus on the generator's efficiency and quality within the G-E framework. To satisfy stringent online latency constraints, a line of Non-Autoregressive (NAR) works, including NAR4Rec\cite{10.1145/3637528.3671645}, NLGR\cite{DBLP:conf/www/WangWKWCTZXW25}, JDRec\cite{DBLP:conf/atal/ZhaoLFGZHCPD24} and OMGRec\cite{DBLP:conf/www/XuXCLSLWZZ26}, 
directly predict positional probabilities for each item in parallel. However, this sacrifices accuracy by overlooking the sequential behavioral dependencies of users browsing the sequence. In contrast, Auto-Regressive (AR) models (e.g., Transformer) deliver superior quality but incur heavy model footprints and step-by-step generation overhead\cite{DBLP:conf/kdd/ZhuLCMSZZX0025}. 
Furthermore, the evaluator-guided AR generator faces a fundamental tension: the evaluator emits a single scalar reward at the end of a length $L$ trajectory, while the generator must commit to $L$ sequential decisions. This induces two coupled difficulties: (i).\textbf{Decoding latency:} Industrial Generator operates under stringent end-to-end latency budgets (typically $time\_cost \le 30$ms total). AR decoding of $L \approx 3\sim 8$ items, even with KV cache, consumes a nontrivial share. (ii).\textbf{Error accumulation:} The teacher-forcing training paradigm feeds ground-truth sequences as inputs, yet inference relies on previously generated subsequences, inevitably accumulating prediction errors. Moreover, such errors amplify as the decoding path extends. 

A natural question arises: can the generation horizon be reduced without sacrificing the expressiveness of the generator? We answer in the affirmative. Our key observation is that the semantic unit of rerank utility is not the individual item but the interactions between items. Co-occurrence effects, complementarity, and list-level diversity are inherently pairwise phenomena; item-level action granularity is an artifact of the autoregressive abstraction, not a property of the reward structure. We thus propose to elevate the generation atom from items to ordered item pairs. In concrete terms, the proposed \textbf{Pair-Space Generation (PSG)} treats each ordered pair $(v_i, v_j) \in [n] \times [n]$ ($i \neq j$) as a generation token. The generator produces $L/2$ pair tokens autoregressively; the resulting sequence is then unfolded deterministically into the original length-$L$ item sequence. Three design components make this practical at industrial scale. First, a pretrained pair encoder computes pair token representations on-the-fly from the per-request candidate set, sidestepping the data sparsity that would otherwise plague a static $n^2$-sized embedding table. Second, a dynamic-vocabulary decoder scores pair tokens via inner product with the encoder output, allowing the same generator to serve arbitrary candidate sets without retraining. Third, a pair-space reinforcement learning algorithm operates natively over $L/2$-step trajectories, with action-mask constraints enforcing item non-repetition and group sampling tuned for the expanded action space. 
Formally, we summarize our contributions:
 \begin{itemize}
    \item We propose Pair-Space Generation, a reformulation of generative rerank that elevates the action granularity from items to ordered pairs while preserving full expressiveness.

     \item  We prove that PSG enjoys a provable nearly $4\times$ reduction in the worst-case sub-optimality bound under outcome-only rewards, owing to the quadratic dependence of error compounding on the decoding horizon, and achieves nearly $2\times$ to $4\times$ inference speedup theoretically by shifting computation from the sequential KV-cache bottleneck to the parallel output-projection path.


     \item  We demonstrate the superiority of PSG against state-of-the-art baselines on public benchmarks, achieving significant and consistent improvements. Beyond offline evaluation, we empirically validate PSG on a production short-video platform (Kuaishou) reranking workload serving over 400 million daily active users, where it delivers a 0.178\% lift in stay-time utility and a $1.83\times$ decoding speedup over a strong item-space G-E baseline.
 \end{itemize}

\section{Related Work}
Our discussion of related work falls into three categories. First, we cover generative recommendation. Second, we examine acceleration techniques for AR paradigms. Third, we review existing efforts on mitigating error accumulation in AR generation.
\subsection{Generative  recommendation}
Retrieval has recently shifted from discriminative to generative paradigms. RecForest\cite{DBLP:conf/nips/FengLLLC22} and DSI\cite{DBLP:conf/nips/Tay00NBM000GSCM22,chen2023understanding} pioneer seq-to-seq retrieval by encoding items into discrete code lists via indexing structures, which are then reconstructed by generative models. TIGER\cite{DBLP:conf/nips/RajputMSKVHHT0S23} adopts Semantic IDs (SIDs),
discrete representations derived from quantization, to denote items, employing generative reconstruction to recover SID sequences. The HSTU\cite{DBLP:conf/icml/ZhaiLLWLCGGGHLS24} 
primarily addresses efficiency and scaling-law properties in generative recommendation, while OneRec\cite{DBLP:journals/corr/abs-2512-24762}
incorporates reward-guided reinforcement learning (e.g., GRPO\cite{DBLP:journals/corr/abs-2402-03300}) to align generated results with user preferences. For sequence-level reranking, generator-only approaches such as Seq2Slate\cite{DBLP:journals/corr/abs-1810-02019}, miRNN\cite{DBLP:conf/ijcai/ZhuangOW18}, List-CVAE\cite{DBLP:conf/iclr/JiangGQMR19}
and SortGen\cite{DBLP:conf/sigir/MengGCLZ25}
directly approximate the optimal sequence by generative models. To improve the robustness, FSC\cite{DBLP:journals/corr/abs-2005-12206} and GRN\cite{DBLP:journals/corr/abs-2104-00860} introduce critics to evaluate generated lists, which marks the emergence of the Generator-Evaluator (G-E) paradigm. Other works, like PIER\cite{DBLP:conf/kdd/ShiYWWGLWWW23}, PRS\cite{DBLP:journals/corr/abs-2102-12057}, GFN4Rec\cite{DBLP:conf/kdd/00060HSMZ0G23}, CGA\cite{DBLP:conf/kdd/ZhuLCMSZZX0025} and YOLOR\cite{10.1145/3746252.3761539}, follow this G-E paradigm, driving continued improvements in both efficiency and accuracy. Specifically, among existing methods, GoalRank\cite{zhang2026goalrank} is most similar to the proposed PSG: it also employs an evaluator to guide generation via reinforcement learning, yet operates in item space rather than pair-token space under the same Generator-Evaluator paradigm.

\subsection{Acceleration of auto-regressive generation}
Typical AR generation decodes step-by-step with complexity linear to sequence length. To meet stringent latency, extensive efforts have focused on accelerating decoding. \textbf{Decoder Architecture:} Speculative decoding\cite{DBLP:conf/emnlp/Xia0WCWS23,DBLP:conf/icml/LeviathanKM23,DBLP:journals/corr/abs-2302-01318} verifies multi-token predictions against AR outputs, reducing forward passes and adopted in LLMs.
Similar techniques appear in generative recommendation (RPG\cite{DBLP:conf/kdd/Hou0SJSSHYM25}, HiCoGen\cite{DBLP:conf/cikm/Li0CD0HLZ25}, GReF\cite{DBLP:conf/cikm/LinLDBLYZ025}, OneRanker\cite{DBLP:journals/corr/abs-2603-02999}). N-gram\cite{DBLP:journals/bstj/Shannon48} assumes dependence on the preceding $n-1$ states, shortening context and cutting attention costs\cite{DBLP:conf/naacl/OuCT24,DBLP:conf/emnlp/ChenLLQGAZW25}.
\textbf{Efficient Attention:} PagedAttention\cite{DBLP:conf/sosp/KwonLZ0ZY0ZS23} eliminates KV cache fragmentation; GQA\cite{DBLP:conf/emnlp/AinslieLJZLS23} reduce KV size via shared heads; FlashAttention\cite{DBLP:conf/iclr/Dao24,DBLP:conf/nips/DaoFERR22} accelerates via tiling. Constrained decoding cuts softmax overhead for closed-set tasks\cite{DBLP:conf/cikm/LinLDBLYZ025}. In generative recommendation, HSTU\cite{DBLP:conf/icml/ZhaiLLWLCGGGHLS24} adopts M-FALCON KV cache for faster inference. Additionally, architecture and attention optimizations are compatible. LazyDecoder (OneRec-V2)\cite{DBLP:journals/corr/abs-2508-20900} and NEZHA\cite{DBLP:conf/www/WangZLLLWZLSWXZ26} 
replace the deep encoder with a lightweight processor using lazy cross-attention and KV sharing. Different from these module-level modifications, our PSG offers a complementary direction: it shortens the decoding horizon by packing multiple items into one token, and can be seamlessly integrated with all existing acceleration techniques.
\subsection{Accumulative error reduction in AR}
Exposure bias, the distribution mismatch between teacher-forcing training and free-running inference, is a primary cause of error accumulation in autoregressive generation. Existing solutions fall into three categories. Training-side approaches bridge the gap via curriculum-based sampling\cite{2969239.2969370}, adversarial distribution matching\cite{DBLP:conf/nips/GoyalLZZCB16}, sequence-level objectives\cite{shen-etal-2016-minimum}, or reinforcement learning\cite{3298483.3298649}. Decoding side strategies, such as beam search, retain multiple candidates to curb error propagation without retraining. Architectural innovations, including non-autoregressive models\cite{schmidt-etal-2022-non}, eliminate causal dependencies to prevent error accumulation structurally. Recent advances extend this direction further: self-forcing training achieves full train-inference alignment\cite{DBLP:conf/nips/HuangLHZS25}; process reward models\cite{DBLP:conf/iclr/LightmanKBEBLLS24} enable step-level error detection and test-time pruning for hierarchical AR generation\cite{guo2026promise}; iterative refinement and masked generation paradigms\cite{DBLP:conf/cvpr/ChangZJLF22} offer alternatives to left-to-right decoding. Our PSG is orthogonal to these efforts: it operates at the representation level by mapping items into pair tokens, shortening the generation horizon without sacrificing expressiveness.

\begin{figure*}[htbp]
\centering
\includegraphics[width=\linewidth]{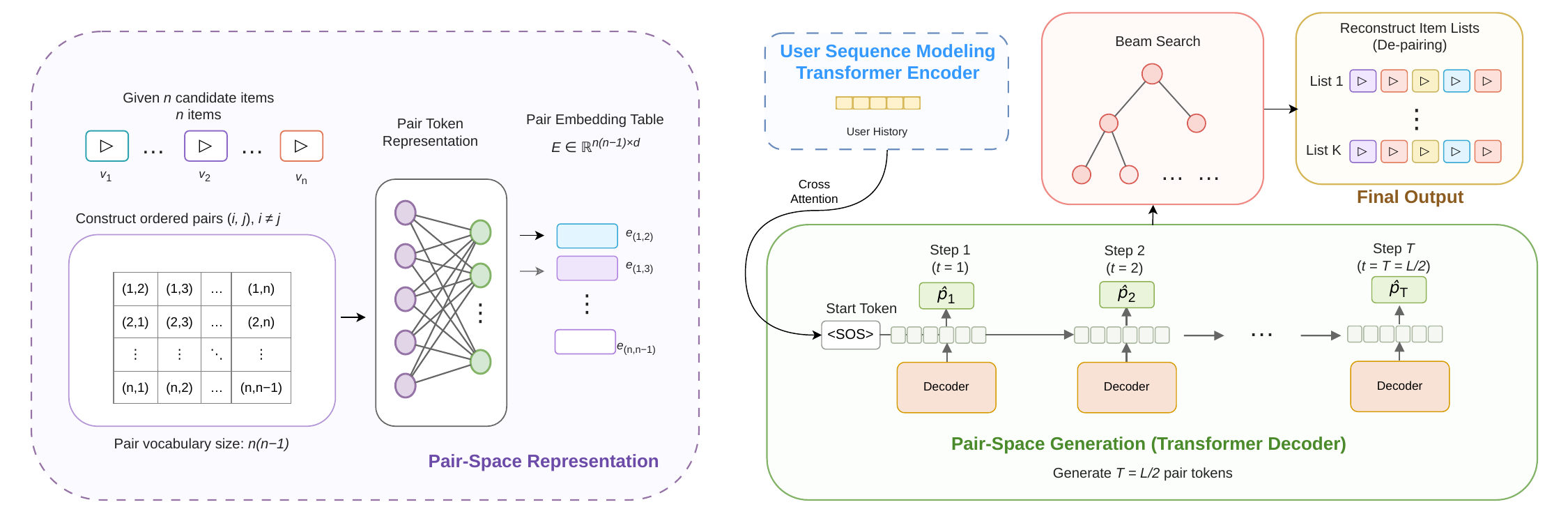}
\caption{PSG: Left part denotes the Pair-Token Representation and right part denotes the Transformer Enc-Dec module.}
\label{fig:module_framework}
\end{figure*}
\section{Method}
In this section, we begin with the problem formulation. We then describe the mapping from item space to pair-token space. Finally, we present the model architecture and training methodology.
\subsection{Reranking formulation}
At each request, given a user context $\bm{u}$ (e.g. user click logs) and a candidate set $\mathcal{V} = \{v_1, \ldots, v_n\}$ of size $n$, let $\Pi \subset \mathcal{V}^{L}$ denote the set of all ordered permutations of $L$ items from $\mathcal{V}$, with $|\Pi| = P(n,L)=\frac{n!}{(n-L)!}$. The reranking task is to produce an ordered list $\bm{\pi} = (\pi_1, \ldots, \pi_L)\in \Pi$ that maximizes the list-level utility $R( \bm{u}, \bm{\pi})$:
$$
    \bm{\pi}^{*}=\mathop{\arg\max}_{\bm{\pi} \in \Pi} R(\bm{u},\bm{\pi}).
$$

Following recent generative recommendation paradigms, we adopt the Generator-Evaluator (G-E) framework for reranking. Concretely, given $\bm{u}$ and $\mathcal{V}$, a generator $\mathcal{G}(\bm{\pi}|\bm{u},\mathcal{V})$ produces $m$ candidate permutations of length $L$. The evaluator $R(\bm{u},\cdot)$ scores each permutation at the sequence level, and the highest-scoring one is selected for exposure. The Auto-Regressive (AR) architecture (e.g., Transformer\cite{3295222.3295349}) fits this generator naturally, producing the permutation item by item:
\begin{equation}
\mathcal{G}_{\bm{\theta}}(\bm{\pi} | \bm{u}, \mathcal{V}) = \prod_{t=1}^L p_{\bm{\theta}}(\pi_t | \pi_{<t}, \bm{u}, \mathcal{V}),
\label{eq:item_decoding_prob_accu}
\end{equation}
where $\bm{\theta}$ denotes the learnable parameters and $p_{\bm{\theta}}(\cdot)$ is the predicted probability at each decoding step. As we focus primarily on the generator, we treat the evaluator $R(\cdot,\cdot)$ as an external plug-in module that provides supervisory signals to guide generation. Our framework is agnostic to the specific form of the evaluator; any sequence-level scoring function can be readily incorporated.

Despite AR's success in LLMs and recommendation systems, it suffers from linearly increasing latency and accumulated decoding errors along the decoding horizon, as discussed earlier. To address these issues, we propose PSG, which shortens the decoding path without sacrificing expressiveness while further reducing decoding errors for reranking.
\subsection{Tokenization in Pair-Space}
We assume $L$ is even (odd $L$ is handled by a trailing single-item slot; see \textbf{Appendix \ref{appendix:odd_L}}). We map the item space into pair-token space by combining two items into one token per request. Specifically, for a candidate set $\mathcal{V}$, we define the per-request pair-token vocabulary
$$
\mathcal{P}(\mathcal{V}) = \{(v_i, v_j) : i, j \in [n], i \neq j\}, \quad |\mathcal{P}(\mathcal{V})| = n(n-1).
$$
The decoding process in the item space of size $n$ is thus transferred to decoding in the pair-token space of size $n(n-1)$ over $L/2$ steps. The superscript $\bm{item}$ or $\bm{pair}$ is used to distinguish the corresponding solution space. In principle, any $k$ items can be combined into one token, yielding a token space of size $\frac{n!}{(n-k)!}$. Specifically, A choice of $k=2$ is adopted for our online deployment setting. The discussion of the choice of $k$ under various constraints is deferred to later sections.

A natural concern is whether this reformulation sacrifices any expressive power. We address this with one theoretical guarantee, collectively certifying PSG as a strict structural refactoring of item-space generation rather than an approximation. First, for any permutation $\bm{\pi}^{\text{item}} \in \Pi^{\text{item}}$, there exists a corresponding $\bm{\pi}^{\text{pair}} \in \Pi^{\text{pair}}$ that unfolds to the same item sequence, and vice versa. Hence, the mapping between $\Pi^{\text{item}}$ and $\Pi^{\text{pair}}$ is bijective. Second, for any item-space generator $\mathcal{G}^{\text{item}}$, there exists a pair-space generator $\mathcal{G}^{\text{pair}}$ that produces identical output distributions. Constructively, define
$$
\begin{aligned}
p((v_{i}, v_{j}) \mid \pi^{\text{pair}}_{<t},\bm{u},\mathcal{P}(\mathcal{V}))
&= p(v_i\mid\pi^{\text{item}}_{<2t-1},\bm{u},\mathcal{V}) \\
&\quad \cdot p(v_{j} \mid \pi^{\text{item}}_{<2t-1}\cup \{v_i\},\bm{u},\mathcal{V}).
\end{aligned}
$$
The product over $t=1,\ldots,L/2$ telescopes to
$$
\prod_{t=1}^{L/2} p(\pi^{\text{pair}}_t | \pi^{\text{pair}}_{<t}, \bm{u},\mathcal{P}(\mathcal{V})) = \prod_{s=1}^{L} p(\pi^{\text{item}}_s | \pi^{\text{item}}_{<s}, \bm{u},\mathcal{V}).
$$
The reverse direction follows symmetrically.
\begin{Corollary}[No expressiveness loss]
\label{corollary:bijection}
The set of list distributions realizable by pair-space autoregressive policies equals that of item-space policies.
\end{Corollary}
\subsection{Model architecture}
The PSG framework consists of two core modules: Pair-Token Representation and Token-Level Generator. The overall architecture is depicted in Figure \ref{fig:module_framework} and will be elaborated in what follows.

\paragraph{Pair-Token Representation:} Each item $v$ is first mapped into a vector embedding $\bm{e}(v) \in \mathbb{R}^{d'}$ from multimodal features (content embeddings, statistical features, taxonomy IDs). Pair-token representations are then computed by a position-aware fusion module:
\begin{equation}
    \bm{e}(v_i, v_j) = \mathrm{MLP}\big([\bm{e}(v_i) + \bm{p}_1;\; \bm{e}(v_j) + \bm{p}_2]\big) \in \mathbb{R}^{d},
    \label{eq:pair_token_emb}
\end{equation}
where $\bm{p}_1, \bm{p}_2 \in \mathbb{R}^{d'}$ are learnable role embeddings, ensuring that $\bm{e}(v_i, v_j) \neq \bm{e}(v_j, v_i)$ ($i\neq j$) in general.

\paragraph{Token-Level Generator: }A standard Transformer Encoder-Decoder architecture is adopted. In RecSys, user behaviors serve as the most critical user features; accordingly, the user $\bm{u}$ is represented via their sequential behavior history. Specifically, the behavior sequence $\bm{X}$ is constructed as:
$$
    \bm{X} = [\bm{x}_1; \bm{x}_2; \cdots; \bm{x}_b] \in \mathbb{R}^{b \times d},
$$
where each $\bm{x}_i$ corresponds to the $i$-th historical behavior. Each behavior representation $\bm{x}_i$ is derived by fusing the interacted item embedding with contextual action signals (e.g., dwell time, interaction type, position) through an MLP: $\bm{x}_i = \mathrm{MLP}([\bm{e}_i; \bm{a}_i])$,
where $\bm{e}_i \in \mathbb{R}^{d_e}$ denotes the item embedding and $\bm{a}_i \in \mathbb{R}^{d_a}$ encodes the action signals (e.g., click type). The encoder maps $\bm{X}$ to a hidden representation via self-attention followed by a position-wise feed-forward network (FFN). Formally, given $\bm{X}$:
\begin{equation}
    \mathrm{SelfAttention}(\bm{Q}, \bm{K}, \bm{V}) = \mathrm{softmax}\left(\frac{\bm{Q}\bm{K}^\top}{\sqrt{d}}\right)\bm{V},
\end{equation}
where $\bm{Q} = \bm{X}\bm{W}_Q$, $\bm{K} = \bm{X}\bm{W}_K$, $\bm{V} = \bm{X}\bm{W}_V$, with $\bm{W}_Q, \bm{W}_K, \bm{W}_V \in \mathbb{R}^{d \times d}$ being learnable projections. The self-attention output is then transformed by a two-layer FFN with ReLU activation. The encoder output is $\mathrm{Enc}(\bm{u}) \in \mathbb{R}^{b \times d}$, where $b$ is the number of historical behaviors. For the decoder, let $<\text{SOS}>$ denote the start token. During training, the decoder input is $\{<\text{SOS}>, \pi^{\text{pair}}_1, \ldots, \pi^{\text{pair}}_{L/2-1}\}$, where $\bm{\pi}^{\text{pair}}$ is the pair-token sequence from the exposed permutation. The decoder hidden states serve as queries $\bm{Q}$, while the encoder output serves as keys $\bm{K}$ and values $\bm{V}$. The cross-attention is computed as:
\begin{equation}
    \mathrm{CrossAttention}(\bm{Q}, \bm{K}, \bm{V}) = \mathrm{softmax}\left(\frac{\bm{Q}\bm{K}^\top}{\sqrt{d}}\right)\bm{V},
\end{equation}
where $\bm{Q}$ is linearly projected from decoder states, and $\bm{K}, \bm{V}$ are derived from $\mathrm{Enc}(\bm{u})$. The cross-attention output is also passed through an FFN. The final decoder output is:
\begin{equation}
    \bm{H} = \mathrm{EncDec}(\bm{u}, \bm{\pi}^{\text{pair}}) = [\bm{h}_1; \ldots; \bm{h}_{L/2}] \in \mathbb{R}^{\frac{L}{2} \times d}.
    \label{eq:decoder_output}
\end{equation}

\subsection{Training}
\paragraph{Pair-Token Pretrain.}To enhance the representation ability of pair-token, PSG pretrains the token embeddings. Given the exposure permutation $\bm{\pi}=[\pi_1,\ldots,\pi_{L}]$ and the corresponding action (e.g. click) $\bm{y}_{\bm{\pi}}=[y_1,\ldots,y_{L}]\in \{0,1\}^{L}$, we construct the pretrain instances as follows.
First, we construct the pair-token by the exposure time, namely $S=\{(\pi_{i},\pi_{j})\}|i<j$. Each pair of $S$ has a pair-level label: if $y_i =1$ and $y_j =1$, $y(\pi_i,\pi_j)=1$; if either $y_i$ or $y_j$ equals one, $y(\pi_i,\pi_j)=0.5$; otherwise $y(\pi_i,\pi_j)=0$, the point-wise average of pair-token embedding(i.e. eq(\ref{eq:pair_token_emb})) approximate the pair-level label trained by mean-square-error. Formally
$$L_{pretrain}=\frac{2}{n(n-1)}\sum_{i<j}\left [\sigma \left(\frac{1}{d}\sum_{t=1}^de(\pi_i,\pi_j)_t\right)-y(\pi_i,\pi_j)\right ]^2,$$
where $\sigma(\cdot)$ means sigmoid function.

\paragraph{Next-Token Prediction.}The next-token prediction (NTP) paradigm is adopted to maintain consistency with training objectives commonly used in online recommender systems. At decoding step $t \in \{1, \ldots, L/2\}$, the decoder output $\bm{h}_t$ from Eq.~(\ref{eq:decoder_output}) is scored against all valid pair-tokens via inner product, followed by a softmax operation to yield the sampling probability over the token vocabulary. Each decoding step naturally corresponds to a multi-classification task, with the training objective defined as:
$$
\mathcal{L}_{\mathrm{NTP}} = -\frac{2}{L} \sum_{t=1}^{L/2} \log \frac{\exp\left(\bm{h}_t^\top \mathrm{emb}(\pi^{\text{pair}}_t)\right)}{\sum_{j=1}^{n(n-1)} \exp\left(\bm{h}_t^\top \mathrm{emb}(\mathrm{token}_j)\right)},
$$
\paragraph{Exploration by reinforcement} The pretrain and NTP task mainly address the exploitation of online logs. Here, we utilize the evaluator $R(\cdot,\cdot)$ to guide the generator by reinforcement, which improves the consistency and exploration ability. We adopt
Group Relative Policy Optimization (GRPO) \cite{DBLP:journals/corr/abs-2402-03300} which exhibits
excellent performance in LLM tasks. First, we sample a group permutation $S=\{\bm{s}_1,\ldots,\bm{s}_{G}\}$ of size $G$ by the generator $\mathcal{G}_{\theta_{old}}$, where $\theta_{old}$ means the former parameters. Then call the evaluator to estimate the corresponding reward $\bm{r}=\{r_1,\ldots,r_G\}$. The relative advantage for each sample can be calculated by
$$A_i=\frac{r_i-avg(\bm{r})}{std(\bm{r})+\epsilon}\ \ \ 1\le i\le G,$$
where $avg(\cdot)$ and $std(\cdot)$ mean average and standard deviation over reward vector respectively. A small $\epsilon$ guarantees non-zero denominator for numerical stability. Then the GRPO loss can be formed as
\begin{equation}
\begin{split}
    L_{GRPO}=\frac{1}{G}&\sum_{i=1}^{G}\min\left(\frac{\pi_{\theta}}{\pi_{\theta_{old}}}\cdot A_i,clip\left(\frac{\pi_{\theta}}{\pi_{\theta_{old}}},1-\delta,1+\delta \right)\right)\\ -
    &\beta D_{KL}(\pi_{\theta}||\pi_{\theta_{old}}),
\end{split}
\end{equation}
where $clip()$ means clipping operation and $D_{KL}()$ means the Kullback-LeiBler distance which avoids too much shift of current policy and the former policy. $\beta$ is a tunable hyper-parameter.

Finally, the training loss is 
$$
L_{loss}=L_{pretrain}+\lambda _{1}\cdot L_{ntp}+\lambda_{2}\cdot L_{GRPO},
$$
where $\lambda_{1}$ and $\lambda_{2}$ are tunable hyper-parameters to control the influence of each task.

\begin{figure}[htbp]
\centering
\includegraphics[width=\columnwidth]{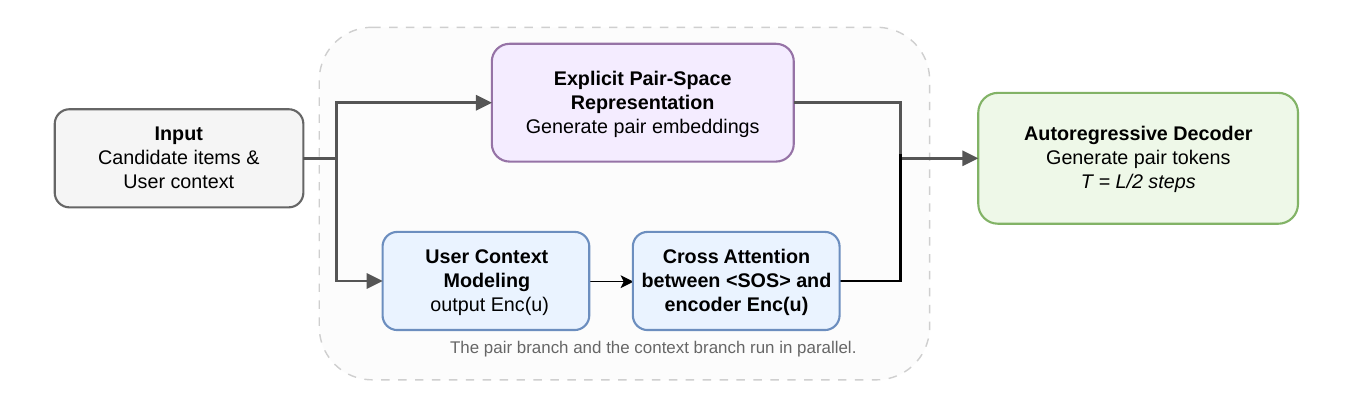}
\caption{Workflow of inference.}
\label{fig:work_flow}
\end{figure}
\section{Analysis for inference}
During inference, at each decoding step, the decoder hidden state is scored against all valid pair-tokens via inner product, followed by softmax to obtain token probabilities. Beam search is then applied to generate candidate pair-token sequences, which are subsequently unfolded into item lists of length $L$.
\subsection{Complexity analysis}
Compared with item-space generation, PSG introduces an extra Pair-Token Representation (PTR) module (Figure~\ref{fig:module_framework}) that maps item pairs into token embeddings. Crucially, both the user-encoder computation and the first decoder step (cross-attention with $<\text{SOS}>$) are independent of the PTR output, the decoder requires $\mathrm{Enc}(\bm{u})$ for cross-attention but not the pair embeddings until output projection at step~1. Since the user encoder and PTR operate on disjoint inputs, they execute concurrently on separate streams (Figure~\ref{fig:work_flow}), and the PTR overhead is absorbed within the user-encoder latency window. Even when executed serially, the PTR accounts for only around 3\% of total inference latency, rendering it negligible. Therefore, in the following complexity analysis, we omit the PTR module. We formalize this argument below and validate it empirically in Section~\ref{Inference Efficiency}.
\begin{proposition}[Pair Encoder Wall-Clock Absorption]
The Pair-Token Representation (PTR) and the user context encoder are data-independent and thus execute concurrently. PTR is a single batch GEMM with arithmetic intensity $\gg 1$, making it compute-bound with near-peak tensor-core throughput \cite{10.1145/1498765.1498785}; the user context encoder comprises attention kernels with arithmetic intensity $\approx 1$, which is memory-bandwidth-bound \cite{DBLP:conf/nips/DaoFERR22}. For $n \le 400$ and $b \ge 50$, the PTR completes within the user encoder window plus Decoder Step 1, which requires $\mathrm{Enc}(\bm{u})$ for cross-attention. Hence, the PTR adds no marginal wall-clock latency and is excluded from the complexity analysis.
\end{proposition}
A notable property of the Transformer decoder is that its parameter count, comprising QKV projection matrices, FFN weights, and layer norms, is independent of the input sequence length. This independence ensures a fair complexity comparison between item-space and pair-space generation under identical model capacity.

Beyond the decoder, overall time complexity is also affected by the encoder. In cases where the encoder processes a very long user history, it may dominate the entire procedure and become the latency bottleneck. Conversely, when the encoder is lightweight, halving the decoding horizon from $L$ to $L/2$ yields a near-$2\times$ reduction in FLOPs for the weight-matrix and cross-attention terms, and a $4\times$ reduction for the KV-cache read term. Since the KV-cache term is memory-bandwidth-bound and accounts for a disproportionate share of wall-clock latency, the overall decode speedup falls between $2\times$ and $4\times$, reaching $1.83\times$ in real deployment (Section~\ref{sec:experiment}). A detailed breakdown of each computational component is provided in the theorem below; the complete proof is deferred to Appendix~\ref{proof:time_complexity}.
\begin{theorem}[Generation Complexity]
\label{thm:complexity}
Let the user encoder be an $H_{\mathrm{enc}}$-layer transformer with hidden size $d$ and FFN width $d_{\mathrm{ff}}$, operating on a user-history sequence of length $b$. The decoder is an $H$-layer transformer with hidden size $d$ and FFN width $d_{\mathrm{ff}}$, operating over a per-request vocabulary of size $|\mathcal{V}|$. At each autoregressive step $t$, the decoder performs self-attention over its own KV cache (length $t$) and cross-attention over the user-encoder output $\mathbf{H}_{\mathrm{enc}} \in \mathbb{R}^{b \times d}$, then projects its hidden state $h_t \in \mathbb{R}^d$ over the vocabulary. The total per-request FLOPs decompose as
\begin{align}
T &= T_{\mathrm{user}} + T_{\mathrm{dec}}(S) \label{eq:total} \\
T_{\mathrm{user}} &= O\!\big(H_{\mathrm{enc}}(d^2 + d \cdot d_{\mathrm{ff}}) \cdot b + H_{\mathrm{enc}} \cdot d \cdot b^2\big) \label{eq:user} \\
T_{\mathrm{dec}}(S) &= \underbrace{O\!\big(H(d^2 + d \cdot d_{\mathrm{ff}} + b \cdot d) \cdot S\big)}_{T_{\mathrm{fixed}}:\;\text{weight-matrix + cross-attn (per-step fixed)}} \nonumber \\
&\quad + \underbrace{O(H \cdot d \cdot S^2)}_{T_{\mathrm{kv}}:\;\text{KV-cache read (sequential)}} + \underbrace{O(|\mathcal{V}| \cdot d \cdot S)}_{T_{\mathrm{proj}}:\;\text{output projection (parallel)}} \label{eq:dec}
\end{align}
The user-encoder cost $T_{\mathrm{user}}$ is identical in item space and pair space (computed once per request). The decoder cost differs by generation horizon $S$ and vocabulary size $|\mathcal{V}|$:
\begin{table}[h]
\centering
\caption{Per-request FLOPs: item space vs.\ pair space.}
\small
\setlength{\tabcolsep}{3.5pt}
\begin{tabular}{lccc}
\toprule
Term & Item ($S\!=\!L$, $|\mathcal{V}|\!=\!n$) & Pair ($S\!=\!L/2$, $|\mathcal{V}|\!=\!n^2$) & Factor \\
\midrule
$T_{\mathrm{fixed}}$ & $O(H(d^2\!+\!dd_{\mathrm{ff}}\!+\!bd)L)$ & $O(H(d^2\!+\!dd_{\mathrm{ff}}\!+\!bd)L/2)$ & $2\times\;\downarrow$ \\
$T_{\mathrm{kv}}$ & $O(HdL^2)$ & $O(HdL^2/4)$ & $\mathbf{4\times}\;\downarrow$ \\
$T_{\mathrm{proj}}$ & $O(ndL)$ & $O(n^2dL/2)$ & $n/2\times\;\uparrow$ \\
\bottomrule
\end{tabular}
\end{table}

\end{theorem}

\subsection{Decoding Error reduction}
Intuitively, shortening the decoding horizon reduces error accumulation, yet expanding the vocabulary from $n$ to $n^2$ may increase per-step error due to both a larger action space and training over $n^2$ classification labels. The net effect is governed by the ratio $\bar{\epsilon}_{\mathrm{item}} / \bar{\epsilon}_{\mathrm{pair}}$: the $4\times$ gain from horizon reduction is partially offset by any increase in per-step mismatch. Under the encoder-saturated regime ($\bar{\epsilon}_{\mathrm{pair}} \approx \bar{\epsilon}_{\mathrm{item}}$), which holds for typical deployment configurations, PSG achieves a nearly $4\times$ error reduction. A formal statement is given below, with the proof deferred to \textbf{Appendix~\ref{proof:error_reduction}}.

\begin{theorem}[Error Compounding in Autoregressive Decoding]
\label{thm:error}
Consider an autoregressive rerank policy generating a length-$L$ list. At each step $t$, the policy's action distribution deviates from the optimal distribution by at most $\bar{\epsilon}$ in total variation. In item space, the per-step vocabulary size is $n$; in pair space, the per-step vocabulary size is $n(n\!-\!1) \approx n^2$. Denote the per-step mismatches as $\bar{\epsilon}_{\mathrm{item}}$ and $\bar{\epsilon}_{\mathrm{pair}}$ respectively, where
\[
\bar{\epsilon}_{\mathrm{pair}} = \bar{\epsilon}_{\mathrm{base}} + \bar{\epsilon}_{\mathrm{enc}}(2) + \bar{\epsilon}_{\mathrm{cov}}(2)
\]
accounts for three error channels: (i) $\bar{\epsilon}_{\mathrm{base}}$, baseline training noise (weakly dependent on vocabulary size); (ii) $\bar{\epsilon}_{\mathrm{enc}}(2)$, pair-encoder generalization error on $n^2$ representations; (iii) $\bar{\epsilon}_{\mathrm{cov}}(2)$, coverage error from finite GRPO group size $G$ over $n^2$ actions.

The trajectory-level and reward-level error bounds are:

\begin{table}[h]
\centering
\caption{Error compounding: item space vs.\ pair space.}
\small
\setlength{\tabcolsep}{5pt}
\begin{tabular}{lcc}
\toprule
Bound & Item Space ($H = L$) & Pair Space ($H = L/2$) \\
\midrule
Trajectory TV mismatch & $L \cdot \bar{\epsilon}_{\mathrm{item}}$ & $(L/2) \cdot \bar{\epsilon}_{\mathrm{pair}}$ \\
Reward sub-optimality & $O(R_{\max} L^2 \bar{\epsilon}_{\mathrm{item}})$ & $O(R_{\max} (L/2)^2 \bar{\epsilon}_{\mathrm{pair}})$ \\
\bottomrule
\end{tabular}
\end{table}

The net improvement ratio is
\[
\frac{\mathrm{SubOpt}_{\mathrm{item}}}{\mathrm{SubOpt}_{\mathrm{pair}}} = 4 \cdot \frac{\bar{\epsilon}_{\mathrm{item}}}{\bar{\epsilon}_{\mathrm{pair}}}.
\]
When $\bar{\epsilon}_{\mathrm{pair}} \leq \bar{\epsilon}_{\mathrm{item}}$ (encoder-saturated regime), the full $4\times$ dividend is realized. When $\bar{\epsilon}_{\mathrm{pair}} > \bar{\epsilon}_{\mathrm{item}}$, the dividend is partially offset. For the deployment regime $n \leq 400$, pretrain density $\rho_2 \gg 10^3$, and $G = 32$: $\bar{\epsilon}_{\mathrm{pair}} \approx \bar{\epsilon}_{\mathrm{item}}$, yielding $\approx 4\times$ improvement.
\end{theorem}

\subsection{Action-Space Operability}
Currently, a token consists of $k=2$ items; here we illustrate this choice is adopted due to $k=2$ get the suitable trade-off between accuracy and efficiency. The vocabulary size is $|\mathcal{P}_k| \approx n^k$, which grows exponentially with $k$. At each decoding step, the generator compute logits $\bm{h}_t^\top E_{\mathrm{vocab}}$ over the full vocabulary for masked softmax, incurring $O(n^k d)$ inner products per step. Meanwhile, the Pair-Token Representation module may exceed the user history encoder in computational cost, potentially becoming a bottleneck for certain combinations of $k$ and $n$ ($n=60$ in our online setting).

Table~\ref{tab:vocab_cost} illustrates the growth pattern. At $k=2$, the vocabulary size is $3{,}540$ for $n=60$ and $160{,}000$ for $n=400$, which remains manageable. At $k=3$, however, the vocabulary grows to $216{,}000$ for $n=60$ and $64\mathrm{M}$ for $n=400$, already straining the sub-30ms latency budget in industrial reranking, particularly given the additional overhead of PTR encoding and beam search decoding. For $k \ge 4$, the vocabulary becomes infeasible (e.g., $\ge 13\mathrm{M}$ for $n=60$ and $\ge 2.6\mathrm{B}$ for $n=400$). In this regime, the computational cost of output projection and PTR encoding dominates, effectively neutralizing the benefit of shortened decoding horizons. We further demonstrate that $k=2$ achieves the best trade-off through empirical studies, with detailed results presented in Section~\ref{sec:k_choose}.
\begin{table}[h]
\centering
\caption{Per-step vocabulary size and output-projection cost.}
\begin{tabular}{cccc}
\toprule
$k$ & $n = 60$ & $n = 400$ & Verdict \\
\midrule
1 & 60 & 400 & Trivial \\
2 & 3,540 & 160k & Comfortable \\
3 & 216k & 64M & Borderline / Infeasible \\
$\geq 4$ & $\geq 13\text{M}$ & $\geq 2.6\text{B}$ & Infeasible \\
\bottomrule
\end{tabular}
\label{tab:vocab_cost}
\end{table}

\section{Experiment}
\label{sec:experiment}
\begin{table}[t]
    \centering
    \caption{Statistics of the datasets used in the experiments.}
    \label{tab:dataset_stats}
    \resizebox{\columnwidth}{!}{
    \begin{tabular}{lcccc}
        \toprule
        Dataset & \# Requests & \# Items & Candidate Pool Size \\
        \midrule
        ML-1M & 161,646 & 3,043 & 50 \\
        Amazon-Books & 309,917 & 38,121 & 50 \\
        RecFlow & 3,308,233 & 14,181,768 & 120 \\
        \bottomrule
    \end{tabular}
    }
\end{table}

\begin{table*}[t]
\centering
\caption{Performance comparison on three datasets.
N@6, P@6, R@6 denotes NDCG@6, Precision@6 and Recall@6.}
\label{tab:main_results}
\setlength{\tabcolsep}{3.2pt}
\resizebox{\textwidth}{!}{
\begin{tabular}{ll*{3}{cccc}}
\toprule
\multirow{2}{*}{Category}
& \multirow{2}{*}{Model}
& \multicolumn{4}{c}{ML-1M}
& \multicolumn{4}{c}{Amazon-Books}
& \multicolumn{4}{c}{RecFlow} \\
\cmidrule(lr){3-6}
\cmidrule(lr){7-10}
\cmidrule(lr){11-14}
& & N@6 & P@6 & R@6 & F1@6
  & N@6 & P@6 & R@6 & F1@6
  & N@6 & P@6 & R@6 & F1@6 \\
\midrule

\multirow{6}{*}{\shortstack{Generator-\\Only}}
& DNN
& 0.5950 & 0.4539 & 0.5542 & 0.4876
& 0.6448 & 0.5072 & 0.6125 & 0.5472
& 0.1584 & 0.0793 & 0.2069 & 0.1084 \\

& DCN
& 0.5981 & 0.4561 & 0.5573 & 0.4901
& 0.6683 & 0.5298 & 0.6461 & 0.5701
& 0.1597 & 0.0795 & 0.2083 & 0.1088 \\

& Seq2Slate
& 0.6222 & 0.4867 & 0.5927 & 0.5225
& 0.6952 & 0.5654 & 0.6871 & 0.6078
& 0.1693 & 0.0821 & 0.2134 & 0.1130 \\

& DLCM
& 0.6061 & 0.4643 & 0.5667 & 0.4988
& 0.6597 & 0.5242 & 0.6396 & 0.5641
& 0.1747 & 0.0861 & 0.2240 & 0.1169 \\

& SetRank
& 0.7154 & 0.5720 & 0.6933 & 0.6132
& 0.8014 & 0.6635 & 0.8145 & 0.7156
& 0.1823 & 0.0896 & 0.2344 & 0.1225 \\

& PRM
& 0.7081 & 0.5639 & 0.6843 & 0.6049
& 0.7992 & 0.6603 & 0.8107 & 0.7122
& 0.1840 & 0.0905 & 0.2368 & 0.1238 \\

\midrule

\multirow{5}{*}{\shortstack{Generator-\\Evaluator}}

& PIER
& 0.7146 & 0.5721 & 0.6932 & 0.6134
& 0.7987 & 0.6610 & 0.8120 & 0.7131
& 0.1910 & 0.0935 & 0.2431 & 0.1277 \\

& NAR4Rec
& 0.7348 & 0.5915 & 0.7168 & 0.6342
& 0.8188 & 0.6786 & 0.8351 & 0.7323
& 0.1792 & 0.0880 & 0.2297 & 0.1203 \\

& JDRec
& \underline{0.7399} & 0.5972 & 0.7233 & 0.6402
& \underline{0.8255} & 0.6832 & \underline{0.8409} & 0.7528
& 0.1832 & 0.0898 & 0.2345 & 0.1227 \\

& OMGRec
& 0.7319 & 0.5886 & 0.7131 & 0.6310
& 0.8040 & 0.6642 & 0.8145 & 0.7160
& 0.1866 & 0.0913 & 0.2385 & 0.1247 \\

& GoalRank
& 0.7232 & \underline{0.6135} & \underline{0.7381} & \underline{0.6701}
& 0.8247 & \underline{0.6905} & 0.8302 & \underline{0.7539}
& \underline{0.1981} & \underline{0.0942} & \underline{0.2456} & \underline{0.1285} \\

\midrule

Ours & \textbf{PSG}
& \textbf{0.7548} & \textbf{0.6431} & \textbf{0.7733} & \textbf{0.7022}
& \textbf{0.8652} & \textbf{0.7725} & \textbf{0.8435} & \textbf{0.8063}
& \textbf{0.2147} & \textbf{0.1024} & \textbf{0.2668} & \textbf{0.1396} \\

\cmidrule{1-14}
& \textbf{Improvement}
& \textbf{+2.01\%} & \textbf{+4.83\%} & \textbf{+4.77\%} & \textbf{+4.79\%}
& \textbf{+4.81\%} & \textbf{+11.88\%} & \textbf{+0.31\%} & \textbf{+6.95\%}
& \textbf{+8.38\%} & \textbf{+8.70\%} & \textbf{+8.63\%} & \textbf{+8.64\%} \\

\bottomrule
\end{tabular}
}
\end{table*}

\subsection{Offline Experiments}
\label{sec:offline_experiments}

\subsubsection{Datasets}

We conduct offline experiments on two public datasets, \textbf{ML-1M} ~\cite{harper2016movielens} and \textbf{Amazon-Books}~\cite{he2016ups}, as well as one industrial dataset collected from a short-video platform, namely \textbf{RecFlow}~\cite{liu2024recflow}. Detailed dataset statistics are provided in Table \ref{tab:dataset_stats}.

Since \textbf{ML-1M} and \textbf{Amazon-Books} do not provide request-level candidate pools for standard reranking instances, we first train a Matrix Factorization (MF) model as the retriever. Specifically, the retriever computes relevance scores between each user and all items, from which we select the top-200 items as the retrieval pool. For each training instance, we then randomly sample 50 items from this top-200 pool as candidate items for reranking. The last six interactions in each user’s historical behavior sequence are used as the ground-truth target list for reranking. The maximum historical sequence length is set to 100 for both \textbf{ML-1M} and \textbf{Amazon-Books}. \textbf{RecFlow} dataset directly provides request-level candidate sets and associated features. For each request, we retain 60 candidate items, and the target reranked list is also set to length 6. The maximum historical sequence length for \textbf{RecFlow} is set to 50.

\subsubsection{Baselines}
We compare PSG against two categories of baselines. The first category consists of Generator-Only methods, including pointwise models (\textbf{DNN}~\cite{DBLP:conf/recsys/CovingtonAS16}, \textbf{DCN}~\cite{DBLP:conf/kdd/WangFFW17}), context-aware refinement models (\textbf{DLCM}~\cite{DBLP:conf/sigir/AiBGC18}, \textbf{PRM}~\cite{DBLP:conf/recsys/PeiZZSLSWJGOP19}), and sequence generation models (\textbf{Seq2Slate}~\cite{DBLP:journals/corr/abs-1810-02019}, \textbf{SetRank}~\cite{DBLP:conf/sigir/PangXALCW20}). The second category consists of Generator-Evaluator methods (\textbf{PIER}~\cite{DBLP:conf/kdd/ShiYWWGLWWW23}, \textbf{NAR4Rec}~\cite{10.1145/3637528.3671645}, \textbf{JDRec}~\cite{DBLP:conf/atal/ZhaoLFGZHCPD24}, \textbf{OMGRec}~\cite{DBLP:conf/www/XuXCLSLWZZ26}, \textbf{GoalRank}~\cite{zhang2026goalrank}).

For PSG, we adopt the same listwise evaluator design as GoalRank (i.e., the OCPM module in PIER). For the generator, we employ a Transformer encoder-decoder architecture, where both the encoder and decoder consist of one layer with one attention head. We first warm up PSG with supervised training for 20 epochs, followed by GRPO optimization. In GRPO, the group size is set to 16, where 16 rerank lists are sampled for each input to compute group-relative advantages. The clipping ratio $\epsilon$ is set to 0.2, and the KL penalty coefficient $\beta$ is set to 0.01. The weights of the NTP loss and GRPO loss are set to $\lambda_1=1$ and $\lambda_2=0.1$, respectively. During inference, beam search with a beam size of 4 is employed. Settings of other baselines follow their corresponding papers.

\subsubsection{Main Results}
As shown in Table~\ref{tab:main_results}, PSG achieves the best overall performance across all three datasets. On ML-1M, it improves NDCG@6, Precision@6, Recall@6, and F1@6 over the strongest competing baselines by 2.01\%, 4.83\%, 4.77\%, and 4.79\%, respectively. On RecFlow, PSG achieves consistent gains on all four metrics, with improvements ranging from 8.38\% to 8.70\%. On Amazon-Books, PSG attains the highest NDCG@6, Precision@6, and F1@6, with a marginal improvement in Recall@6 over the best-performing baseline. We attribute the performance gains to three factors. First, the superior results of Generator-Evaluator methods over Generator-Only baselines highlight the importance of evaluator guidance in reranking. Second, PSG benefits from richer pair-level semantics. It jointly determines two consecutive positions as a single token, so the selection of the first position is implicitly conditioned on the second. This provides a one-step look-ahead that item-space generation structurally lacks. Third, reducing the generation horizon from $L$ to $L/2$ directly suppresses error compounding.

\subsubsection{Inference Efficiency}
\label{Inference Efficiency}
\begin{figure}[t]
\centering
\includegraphics[width=\linewidth]{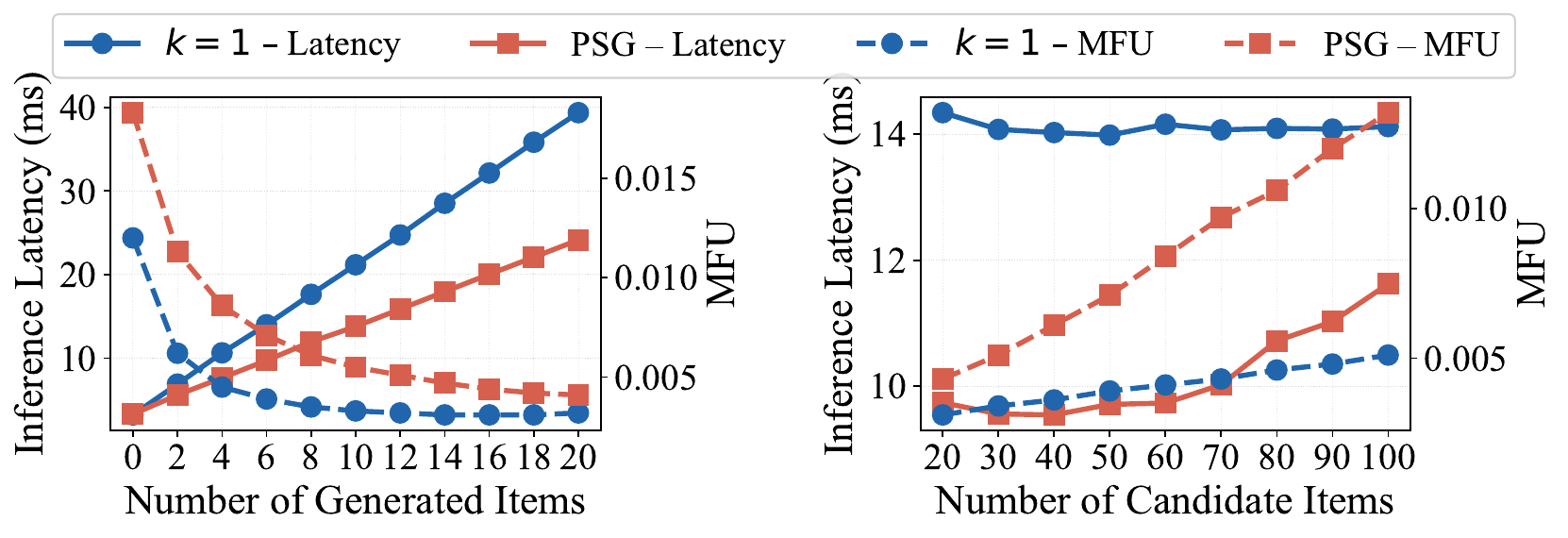}
\caption{Inference latency and hardware utilization of PSG (k = 2) vs. item-space (k = 1).}
\label{fig:latency}
\end{figure}

To better evaluate the inference efficiency advantage of PSG, we additionally train a $k=1$ variant (where each token consists of a single item) and compare it with PSG ($k=2$) in Figure~\ref{fig:latency}. We measure both \textit{Inference Latency} and \textit{MFU} for the two models. All inference experiments are conducted on an
GPU with FP32 precision, which provides a peak computational throughput of 59.8 TFLOPS. The left subfigure studies the impact of the number of generated items under the default setting of 50 candidate items. As expected, the inference latency of both methods increases almost linearly with the number of generated items. PSG consistently achieves lower latency, with the $k=1$ variant being roughly 1.5$\times$ slower on average, while PSG also attains higher MFU (Model FLOPs Utilization) throughout. The right subfigure examines the effect of candidate pool size while fixing the number of generated items to 6. The $k=1$ variant is relatively insensitive to the number of candidates, whereas PSG exhibits a mild increase in latency as the candidate pool size grows, due to the expansion of the pair-token vocabulary. Nevertheless, PSG still maintains a clear latency advantage, with the $k=1$ variant remaining around 1.5$\times$ slower overall. In addition, PSG consistently achieves higher MFU than the $k=1$ variant under different candidate sizes, indicating better hardware utilization during decoding. In Table \ref{tab:ptr_is_negible}, we empirically validate that the time cost of the Pair-Token Representation (PTR) module remains negligible even when executed serially with the Transformer encoder–decoder, accounting for only about 3\% of the total inference latency and thus exerting no measurable impact on the overall runtime.

\begin{table}[t]
\centering
\caption{Pair-Token Representation
(PTR) Overhead Ratio in Total Inference Latency (6 Generated Items)}
\label{tab:pair_overhead}
\resizebox{\columnwidth}{!}{
\begin{tabular}{cccc}
\toprule
Candidate Size $n$ & PTR (ms) & Inference Latency(ms) & Overhead Ratio \\
\midrule
20  & 0.332 & 9.74  & 3.4\% \\
50  & 0.335 & 9.71  & 3.4\% \\
80  & 0.337 & 10.72 & 3.1\% \\
100 & 0.329 & 11.62 & 2.8\% \\
\bottomrule
\end{tabular}
}
\label{tab:ptr_is_negible}
\end{table}

\begin{table}[t]
\centering
\caption{Ablation study of PSG on ML-1M and Amazon-Books.
N@6, P@6, R@6 and F1@6 denote NDCG@6, Precision@6, Recall@6 and F1@6, respectively.}
\label{tab:ablation}
\resizebox{\columnwidth}{!}{
\begin{tabular}{lcccccccc}
\toprule
\multirow{2}{*}{Variant}
& \multicolumn{4}{c}{ML-1M}
& \multicolumn{4}{c}{Amazon-Books} \\
\cmidrule(lr){2-5}
\cmidrule(lr){6-9}
& N@6 & P@6 & R@6 & F1@6
& N@6 & P@6 & R@6 & F1@6 \\
\midrule

PSG
& \textbf{0.7548} & \textbf{0.6431} & \textbf{0.7733} & \textbf{0.7022}
& \textbf{0.8652} & \textbf{0.7725} & \textbf{0.8345} & \textbf{0.8023} \\

w/o Pair Pretrain
& 0.7463 & 0.6329 & 0.7616 & 0.6913
& 0.8527 & 0.7586 & 0.8194 & 0.7878 \\

w/o NTP
& 0.6452 & 0.5243 & 0.6363 & 0.5749
& 0.7595 & 0.6428 & 0.7022 & 0.6712 \\

w/o GRPO
& 0.7232 & 0.6135 & 0.7381 & 0.6701
& 0.8175 & 0.6816 & 0.8132 & 0.7416 \\

\bottomrule
\end{tabular}
}
\label{tab:ablation_study}
\end{table}

\subsubsection{Ablation study}

Our method has three training stages: pair-token pretraining, next-token prediction (NTP), and exploration by reinforcement learning. 
We evaluate their contributions by removing each corresponding objective, demonstrated in Table \ref{tab:ablation_study}.

The ablation results show that all three objectives are beneficial. 
Removing NTP causes the largest performance degradation. This indicates that token-level autoregressive supervision is essential for learning stable sequential item generation.
Pair-token pretraining improves performance by capturing local item co-occurrence patterns as item-level relation priors, which provides better initialization for subsequent autoregressive training.
Removing GRPO also degrades performance, indicating its importance in policy refinement after supervised training. By leveraging group-relative feedback, GRPO further improves the quality of generated recommendation slates beyond sequence imitation.

\subsection{Online A/B test}
\label{sec:online_ab_test}
PSG is additionally deployed in the reranking stage of the single-column recommendation feed on the Kuaishou main app, which serves over 400 million daily active users with an average daily time spent exceeding two hours per user, providing a large-scale and latency-sensitive production environment. The online baseline is GoalRank\cite{zhang2026goalrank}, which takes $60$ candidate items and a user behavior sequence of length $100$ as input, and generates $50$ candidate sequences each of length $L=6$. PSG and GoalRank share all environment settings, including the online evaluator; as the evaluator design is not the focus of this work, its details are omitted. Two disjoint 10\% traffic buckets are allocated for seven days during the online test. Following the platform's latency budget, 50 ms is assigned to the G-E paradigm, with 30 ms for the generator to produce candidate sequences and 20 ms for the evaluator to select the optimal one. The primary online metrics are the Stay Time (ST) per user on the platform. Additionally, latency and online QPS (queries per second) are reported, both measured on the same Cloud Container equipped with 120 CPU threads and 500 GB of memory. These metrics collectively reflect the online efficiency of the system. Results are exhibited in Table \ref{tab:online_result}.
\begin{table}[t]
\centering
\caption{Online comparison GoalRank (baseline) VS PSG.}
\label{tab:online_result}
\begin{tabular}{lccc}
\toprule
Metric & GoalRank & PSG & Improvement \\
\midrule
ST (relative lift) & -- & 0.178\% & -- \\
QPS per cloud container & 734 & 1320 & +79.8\% \\
Generator Latency (ms) & 38.42 & 20.99 & 1.83$\times$ speedup \\
\bottomrule
\end{tabular}
\label{tab:online_result}
\end{table}

Such a relative 0.178\% lift in stay time is particularly significant, as on a web-scale platform like Kuaishou, a 0.1\% improvement is typically the threshold for rolling out a new online model to full traffic. Notably, PSG can handle nearly 80\% additional traffic under the same container configuration, implying substantial savings in computational budget and thus significant economic benefits for the company in supporting full online service. However, the online environment is inherently complex. The latency budget is not solely consumed by model execution; a considerable portion is allocated to request parsing, workflow preparation, logging, response construction, and other overheads. Once the total cost exceeds the generator budget (30 ms), the system terminates the request process. To avoid degrading user experience during testing, we evaluate latency on dry-run traffic, which replicates real online requests but does not affect actual user exposure. Under the same QPS (e.g. 1320 QPS) load on a single cloud container, PSG achieves a $1.83\times$ speedup over GoalRank and successfully handles all incoming requests.

\begin{table}[t]
\centering
\caption{Average inference latency (ms) over 6 generated items 
under different numbers of items per token ($k$) and candidate size ($n$).  
$k\!=\!3$ runs out of memory at $n\!=\!100$ on a 48GB GPU.}
\label{tab:granularity}
\resizebox{\columnwidth}{!}{
\begin{tabular}{ccccccccccc}
\toprule
\multirow{2}{*}{$k$} & \multicolumn{9}{c}{Candidate Size ($n$)} \\
\cmidrule(lr){2-10}
& 20 & 30 & 40 & 50 & 60 & 70 & 80 & 90 & 100 \\
\midrule
1 & 25.3 & 24.9 & 24.8 & 24.9 & 25.0 & 25.0 & 25.1 & 25.0 & 25.0 \\
2 & 15.7 & \textbf{15.5} & \textbf{15.5} & \textbf{15.9} & \textbf{15.9} & \textbf{16.3} & \textbf{17.1} & \textbf{17.8} & \textbf{18.7} \\
3 & \textbf{14.0} & 20.7 & 34.6 & 62.7 & 106.5 & 164.8 & 241.3 & 340.7 & OOM \\
\bottomrule
\end{tabular}
}
\label{tab:k_2_is_reasonable}
\end{table}

\subsection{Pair-Token ($k=2$) is reasonable}
\label{sec:k_choose}
We further examine the $k=3$ variant under the same beam size setting (beam size $=4$) and evaluate its inference latency in Table~\ref{tab:k_2_is_reasonable}. When the candidate set size is relatively small (e.g., $n=20$), the inference time is further reduced due to the shorter decoding sequence. However, as the candidate pool size increases, the computational cost of processing the vocabulary grows rapidly, the per-step vocabulary size expands to $n^3$, and this overhead eventually outweighs the latency reduction brought by fewer decoding steps. These results demonstrate that pair-token modeling with $k=2$ strikes the best balance between decoding efficiency and computational tractability. This is particularly critical for our online environment ($n=60$, CPU-based inference on Cloud Container), where $k=3$ would push the generator beyond the $30$ms budget.

\section{Conclusion}
\label{sec:conclusion}
We presented Pair-Space Generation (PSG), which replaces item-level autoregressive decoding with pair-level decoding for RL-based reranking. PSG is provably as expressive as item-space generation, halves the decoding horizon to yield a near-$2\times$ speedup (empirically $1.83\times$), and reduces the worst-case suboptimality bound by $4\times$ under outcome-only rewards, all confirmed by online A/B tests in a production system serving hundreds of millions of daily active users. A pretrained pair-token representation module with order-aware objectives addresses the cold-start challenge that would otherwise plague an $n^2$-scale vocabulary. Extending PSG to triple-space is theoretically feasible, yet as our analysis shows, $k=2$ remains the industrial sweet spot; tighter instance-dependent error bounds and stochastic transition extensions are left for future work.


\bibliographystyle{ACM-Reference-Format}
\bibliography{references.bib}

@inproceedings{DBLP:conf/kdd/Hou0SJSSHYM25,
  author       = {Yupeng Hou and
                  Jiacheng Li and
                  Ashley Shin and
                  Jinsung Jeon and
                  Abhishek Santhanam and
                  Wei Shao and
                  Kaveh Hassani and
                  Ning Yao and
                  Julian J. McAuley},
  editor       = {Luiza Antonie and
                  Jian Pei and
                  Xiaohui Yu and
                  Flavio Chierichetti and
                  Hady W. Lauw and
                  Yizhou Sun and
                  Srinivasan Parthasarathy},
  title        = {Generating Long Semantic IDs in Parallel for Recommendation},
  booktitle    = {Proceedings of the 31st {ACM} {SIGKDD} Conference on Knowledge Discovery
                  and Data Mining, V.2, {KDD} 2025, Toronto ON, Canada, August 3-7,
                  2025},
  pages        = {956--966},
  publisher    = {{ACM}},
  year         = {2025},
  url          = {https://doi.org/10.1145/3711896.3736979},
  doi          = {10.1145/3711896.3736979},
  bibsource    = {dblp computer science bibliography, https://dblp.org}
}

@inproceedings{DBLP:conf/cikm/Li0CD0HLZ25,
  author       = {Wuchao Li and
                  Shiyao Wang and
                  Kuo Cai and
                  Jiaxin Deng and
                  Xingmei Wang and
                  Qigen Hu and
                  Defu Lian and
                  Guorui Zhou},
  editor       = {Meeyoung Cha and
                  Chanyoung Park and
                  Noseong Park and
                  Carl Yang and
                  Senjuti Basu Roy and
                  Jessie Li and
                  Jaap Kamps and
                  Kijung Shin and
                  Bryan Hooi and
                  Lifang He},
  title        = {Taming Ultra-Long Behavior Sequence in Session-wise Generative Recommendation},
  booktitle    = {Proceedings of the 34th {ACM} International Conference on Information
                  and Knowledge Management, {CIKM} 2025, Seoul, Republic of Korea, November
                  10-14, 2025},
  pages        = {5839--5846},
  publisher    = {{ACM}},
  year         = {2025},
  url          = {https://doi.org/10.1145/3746252.3761564},
  doi          = {10.1145/3746252.3761564},
  bibsource    = {dblp computer science bibliography, https://dblp.org}
}

@inproceedings{DBLP:conf/cikm/LinLDBLYZ025,
  author       = {Zhijie Lin and
                  Zhuofeng Li and
                  Chenglei Dai and
                  Wentian Bao and
                  Shuai Lin and
                  Enyun Yu and
                  Haoxiang Zhang and
                  Liang Zhao},
  editor       = {Meeyoung Cha and
                  Chanyoung Park and
                  Noseong Park and
                  Carl Yang and
                  Senjuti Basu Roy and
                  Jessie Li and
                  Jaap Kamps and
                  Kijung Shin and
                  Bryan Hooi and
                  Lifang He},
  title        = {GReF: {A} Unified Generative Framework for Efficient Reranking via
                  Ordered Multi-token Prediction},
  booktitle    = {Proceedings of the 34th {ACM} International Conference on Information
                  and Knowledge Management, {CIKM} 2025, Seoul, Republic of Korea, November
                  10-14, 2025},
  pages        = {5879--5887},
  publisher    = {{ACM}},
  year         = {2025},
  url          = {https://doi.org/10.1145/3746252.3761540},
  doi          = {10.1145/3746252.3761540},
  bibsource    = {dblp computer science bibliography, https://dblp.org}
}

@article{DBLP:journals/corr/abs-2603-02999,
  author       = {Dekai Sun and
                  Yiming Liu and
                  Jiafan Zhou and
                  Xun Liu and
                  Chenchen Yu and
                  Yi Li and
                  Jun Zhang and
                  Huan Yu and
                  Jie Jiang},
  title        = {OneRanker: Unified Generation and Ranking with One Model in Industrial
                  Advertising Recommendation},
  journal      = {CoRR},
  volume       = {abs/2603.02999},
  year         = {2026},
  url          = {https://doi.org/10.48550/arXiv.2603.02999},
  doi          = {10.48550/ARXIV.2603.02999},
  eprinttype   = {arXiv},
  eprint       = {2603.02999},
  bibsource    = {dblp computer science bibliography, https://dblp.org}
}

@article{DBLP:journals/bstj/Shannon48,
  author       = {Claude E. Shannon},
  title        = {A mathematical theory of communication},
  journal      = {Bell Syst. Tech. J.},
  volume       = {27},
  number       = {3},
  pages        = {379--423},
  year         = {1948},
  url          = {https://doi.org/10.1002/j.1538-7305.1948.tb01338.x},
  doi          = {10.1002/J.1538-7305.1948.TB01338.X},
  bibsource    = {dblp computer science bibliography, https://dblp.org}
}

@inproceedings{DBLP:conf/naacl/OuCT24,
  author       = {Jie Ou and
                  Yueming Chen and
                  Wenhong Tian},
  editor       = {Yi Yang and
                  Aida Mostafazadeh Davani and
                  Avi Sil and
                  Anoop Kumar},
  title        = {Lossless Acceleration of Large Language Model via Adaptive N-gram
                  Parallel Decoding},
  booktitle    = {Proceedings of the 2024 Conference of the North American Chapter of
                  the Association for Computational Linguistics: Human Language Technologies:
                  Industry Track, {NAACL} 2024, Mexico City, Mexico, June 16-21, 2024},
  pages        = {10--22},
  publisher    = {Association for Computational Linguistics},
  year         = {2024},
  url          = {https://doi.org/10.18653/v1/2024.naacl-industry.2},
  doi          = {10.18653/V1/2024.NAACL-INDUSTRY.2},
  bibsource    = {dblp computer science bibliography, https://dblp.org}
}

@inproceedings{DBLP:conf/emnlp/ChenLLQGAZW25,
  author       = {Jinglin Chen and
                  Qiwei Li and
                  Zuchao Li and
                  Baoyuan Qi and
                  Guoming Liu and
                  Haojun Ai and
                  Hai Zhao and
                  Ping Wang},
  editor       = {Christos Christodoulopoulos and
                  Tanmoy Chakraborty and
                  Carolyn Rose and
                  Violet Peng},
  title        = {Faster In-Context Learning for LLMs via N-Gram Trie Speculative Decoding},
  booktitle    = {Proceedings of the 2025 Conference on Empirical Methods in Natural
                  Language Processing, {EMNLP} 2025, Suzhou, China, November 4-9, 2025},
  pages        = {18040--18051},
  publisher    = {Association for Computational Linguistics},
  year         = {2025},
  url          = {https://doi.org/10.18653/v1/2025.emnlp-main.911},
  doi          = {10.18653/V1/2025.EMNLP-MAIN.911},
  bibsource    = {dblp computer science bibliography, https://dblp.org}
}

@inproceedings{DBLP:conf/emnlp/Xia0WCWS23,
  author       = {Heming Xia and
                  Tao Ge and
                  Peiyi Wang and
                  Si{-}Qing Chen and
                  Furu Wei and
                  Zhifang Sui},
  editor       = {Houda Bouamor and
                  Juan Pino and
                  Kalika Bali},
  title        = {Speculative Decoding: Exploiting Speculative Execution for Accelerating
                  Seq2seq Generation},
  booktitle    = {Findings of the Association for Computational Linguistics: {EMNLP}
                  2023, Singapore, December 6-10, 2023},
  series       = {Findings of {ACL}},
  volume       = {{EMNLP} 2023},
  pages        = {3909--3925},
  publisher    = {Association for Computational Linguistics},
  year         = {2023},
  url          = {https://doi.org/10.18653/v1/2023.findings-emnlp.257},
  doi          = {10.18653/V1/2023.FINDINGS-EMNLP.257},
  bibsource    = {dblp computer science bibliography, https://dblp.org}
}

@inproceedings{DBLP:conf/icml/LeviathanKM23,
  author       = {Yaniv Leviathan and
                  Matan Kalman and
                  Yossi Matias},
  editor       = {Andreas Krause and
                  Emma Brunskill and
                  Kyunghyun Cho and
                  Barbara Engelhardt and
                  Sivan Sabato and
                  Jonathan Scarlett},
  title        = {Fast Inference from Transformers via Speculative Decoding},
  booktitle    = {International Conference on Machine Learning, {ICML} 2023, 23-29 July
                  2023, Honolulu, Hawaii, {USA}},
  series       = {Proceedings of Machine Learning Research},
  volume       = {202},
  pages        = {19274--19286},
  publisher    = {{PMLR}},
  year         = {2023},
  url          = {https://proceedings.mlr.press/v202/leviathan23a.html},
  bibsource    = {dblp computer science bibliography, https://dblp.org}
}

@article{DBLP:journals/corr/abs-2302-01318,
  author       = {Charlie Chen and
                  Sebastian Borgeaud and
                  Geoffrey Irving and
                  Jean{-}Baptiste Lespiau and
                  Laurent Sifre and
                  John Jumper},
  title        = {Accelerating Large Language Model Decoding with Speculative Sampling},
  journal      = {CoRR},
  volume       = {abs/2302.01318},
  year         = {2023},
  url          = {https://doi.org/10.48550/arXiv.2302.01318},
  doi          = {10.48550/ARXIV.2302.01318},
  eprinttype   = {arXiv},
  eprint       = {2302.01318},
  bibsource    = {dblp computer science bibliography, https://dblp.org}
}

@inproceedings{DBLP:conf/sosp/KwonLZ0ZY0ZS23,
  author       = {Woosuk Kwon and
                  Zhuohan Li and
                  Siyuan Zhuang and
                  Ying Sheng and
                  Lianmin Zheng and
                  Cody Hao Yu and
                  Joseph Gonzalez and
                  Hao Zhang and
                  Ion Stoica},
  editor       = {Jason Flinn and
                  Margo I. Seltzer and
                  Peter Druschel and
                  Antoine Kaufmann and
                  Jonathan Mace},
  title        = {Efficient Memory Management for Large Language Model Serving with
                  PagedAttention},
  booktitle    = {Proceedings of the 29th Symposium on Operating Systems Principles,
                  {SOSP} 2023, Koblenz, Germany, October 23-26, 2023},
  pages        = {611--626},
  publisher    = {{ACM}},
  year         = {2023},
  url          = {https://doi.org/10.1145/3600006.3613165},
  doi          = {10.1145/3600006.3613165},
  bibsource    = {dblp computer science bibliography, https://dblp.org}
}

@inproceedings{DBLP:conf/emnlp/AinslieLJZLS23,
  author       = {Joshua Ainslie and
                  James Lee{-}Thorp and
                  Michiel de Jong and
                  Yury Zemlyanskiy and
                  Federico Lebr{\'{o}}n and
                  Sumit Sanghai},
  editor       = {Houda Bouamor and
                  Juan Pino and
                  Kalika Bali},
  title        = {{GQA:} Training Generalized Multi-Query Transformer Models from Multi-Head
                  Checkpoints},
  booktitle    = {Proceedings of the 2023 Conference on Empirical Methods in Natural
                  Language Processing, {EMNLP} 2023, Singapore, December 6-10, 2023},
  pages        = {4895--4901},
  publisher    = {Association for Computational Linguistics},
  year         = {2023},
  url          = {https://doi.org/10.18653/v1/2023.emnlp-main.298},
  doi          = {10.18653/V1/2023.EMNLP-MAIN.298},
  bibsource    = {dblp computer science bibliography, https://dblp.org}
}

@inproceedings{DBLP:conf/iclr/Dao24,
  author       = {Tri Dao},
  title        = {FlashAttention-2: Faster Attention with Better Parallelism and Work
                  Partitioning},
  booktitle    = {The Twelfth International Conference on Learning Representations,
                  {ICLR} 2024, Vienna, Austria, May 7-11, 2024},
  publisher    = {OpenReview.net},
  year         = {2024},
  url          = {https://openreview.net/forum?id=mZn2Xyh9Ec},
  bibsource    = {dblp computer science bibliography, https://dblp.org}
}

@inproceedings{DBLP:conf/nips/DaoFERR22,
  author       = {Tri Dao and
                  Daniel Y. Fu and
                  Stefano Ermon and
                  Atri Rudra and
                  Christopher R{\'{e}}},
  editor       = {Sanmi Koyejo and
                  S. Mohamed and
                  A. Agarwal and
                  Danielle Belgrave and
                  K. Cho and
                  A. Oh},
  title        = {FlashAttention: Fast and Memory-Efficient Exact Attention with IO-Awareness},
  booktitle    = {Advances in Neural Information Processing Systems 35: Annual Conference
                  on Neural Information Processing Systems 2022, NeurIPS 2022, New Orleans,
                  LA, USA, November 28 - December 9, 2022},
  year         = {2022},
  url          = {http://papers.nips.cc/paper\_files/paper/2022/hash/67d57c32e20fd0a7a302cb81d36e40d5-Abstract-Conference.html},
  bibsource    = {dblp computer science bibliography, https://dblp.org}
}

@article{DBLP:journals/corr/abs-2508-20900,
  author       = {Guorui Zhou and
                  Hengrui Hu and
                  Hongtao Cheng and
                  Huanjie Wang and
                  Jiaxin Deng and
                  Jinghao Zhang and
                  Kuo Cai and
                  Lejian Ren and
                  Lu Ren and
                  Liao Yu and
                  Pengfei Zheng and
                  Qiang Luo and
                  Qianqian Wang and
                  Qigen Hu and
                  Rui Huang and
                  Ruiming Tang and
                  Shiyao Wang and
                  Shujie Yang and
                  Tao Wu and
                  Wuchao Li and
                  Xinchen Luo and
                  Xingmei Wang and
                  Yi Su and
                  Yunfan Wu and
                  Zexuan Cheng and
                  Zhanyu Liu and
                  Zixing Zhang and
                  Bin Zhang and
                  Boxuan Wang and
                  Chaoyi Ma and
                  Chengru Song and
                  Chenhui Wang and
                  Chenglong Chu and
                  Di Wang and
                  Dongxue Meng and
                  Dunju Zang and
                  Fan Yang and
                  Fangyu Zhang and
                  Feng Jiang and
                  Fuxing Zhang and
                  Gang Wang and
                  Guowang Zhang and
                  Han Li and
                  Honghui Bao and
                  Hongyang Cao and
                  Jiaming Huang and
                  Jiapeng Chen and
                  Jiaqiang Liu and
                  Jinghui Jia and
                  Kun Gai and
                  Lantao Hu and
                  Liang Zeng and
                  Qiang Wang and
                  Qidong Zhou and
                  Rongzhou Zhang and
                  Shengzhe Wang and
                  Shihui He and
                  Shuang Yang and
                  Siyang Mao and
                  Sui Huang and
                  Tiantian He and
                  Tingting Gao and
                  Wei Yuan and
                  Xiao Liang and
                  Xiaoxiao Xu and
                  Xugang Liu and
                  Yan Wang and
                  Yang Zhou and
                  Yi Wang and
                  Yiwu Liu and
                  Yue Song and
                  Yufei Zhang and
                  Yunfeng Zhao and
                  Zhixin Ling and
                  Ziming Li},
  title        = {OneRec-V2 Technical Report},
  journal      = {CoRR},
  volume       = {abs/2508.20900},
  year         = {2025},
  url          = {https://doi.org/10.48550/arXiv.2508.20900},
  doi          = {10.48550/ARXIV.2508.20900},
  eprinttype   = {arXiv},
  eprint       = {2508.20900},
  bibsource    = {dblp computer science bibliography, https://dblp.org}
}

@inproceedings{DBLP:conf/www/WangZLLLWZLSWXZ26,
  author       = {Yejing Wang and
                  Shengyu Zhou and
                  Jinyu Lu and
                  Ziwei Liu and
                  Langming Liu and
                  Maolin Wang and
                  Wenlin Zhang and
                  Feng Li and
                  Wenbo Su and
                  Pengjie Wang and
                  Jian Xu and
                  Xiangyu Zhao},
  editor       = {Hakim Hacid and
                  Yoelle Maarek and
                  Francesco Bonchi and
                  Ido Guy and
                  Emine Yilmaz},
  title        = {{NEZHA:} {A} Zero-sacrifice and Hyperspeed Decoding Architecture for
                  Generative Recommendations},
  booktitle    = {Proceedings of the {ACM} Web Conference 2026, {WWW} 2026, Dubai, United
                  Arab Emirates, originally scheduled for April 13-17, 2026, rescheduled
                  for June 29 - July 3, 2026},
  pages        = {8073--8082},
  publisher    = {{ACM}},
  year         = {2026},
  url          = {https://doi.org/10.1145/3774904.3792797},
  doi          = {10.1145/3774904.3792797},
  bibsource    = {dblp computer science bibliography, https://dblp.org}
}

@inproceedings{DBLP:conf/icml/ZhaiLLWLCGGGHLS24,
  author       = {Jiaqi Zhai and
                  Lucy Liao and
                  Xing Liu and
                  Yueming Wang and
                  Rui Li and
                  Xuan Cao and
                  Leon Gao and
                  Zhaojie Gong and
                  Fangda Gu and
                  Jiayuan He and
                  Yinghai Lu and
                  Yu Shi},
  editor       = {Ruslan Salakhutdinov and
                  Zico Kolter and
                  Katherine A. Heller and
                  Adrian Weller and
                  Nuria Oliver and
                  Jonathan Scarlett and
                  Felix Berkenkamp},
  title        = {Actions Speak Louder than Words: Trillion-Parameter Sequential Transducers
                  for Generative Recommendations},
  booktitle    = {Forty-first International Conference on Machine Learning, {ICML} 2024,
                  Vienna, Austria, July 21-27, 2024},
  series       = {Proceedings of Machine Learning Research},
  volume       = {235},
  pages        = {58484--58509},
  publisher    = {{PMLR} / OpenReview.net},
  year         = {2024},
  url          = {https://proceedings.mlr.press/v235/zhai24a.html},
  bibsource    = {dblp computer science bibliography, https://dblp.org}
}

@inproceedings{DBLP:conf/nips/FengLLLC22,
  author       = {Chao Feng and
                  Wuchao Li and
                  Defu Lian and
                  Zheng Liu and
                  Enhong Chen},
  editor       = {Sanmi Koyejo and
                  S. Mohamed and
                  A. Agarwal and
                  Danielle Belgrave and
                  K. Cho and
                  A. Oh},
  title        = {Recommender Forest for Efficient Retrieval},
  booktitle    = {Advances in Neural Information Processing Systems 35: Annual Conference
                  on Neural Information Processing Systems 2022, NeurIPS 2022, New Orleans,
                  LA, USA, November 28 - December 9, 2022},
  year         = {2022},
  url          = {http://papers.nips.cc/paper\_files/paper/2022/hash/fe2fe749d329627f161484876630c689-Abstract-Conference.html},
  bibsource    = {dblp computer science bibliography, https://dblp.org}
}

@inproceedings{DBLP:conf/nips/Tay00NBM000GSCM22,
  author       = {Yi Tay and
                  Vinh Tran and
                  Mostafa Dehghani and
                  Jianmo Ni and
                  Dara Bahri and
                  Harsh Mehta and
                  Zhen Qin and
                  Kai Hui and
                  Zhe Zhao and
                  Jai Prakash Gupta and
                  Tal Schuster and
                  William W. Cohen and
                  Donald Metzler},
  editor       = {Sanmi Koyejo and
                  S. Mohamed and
                  A. Agarwal and
                  Danielle Belgrave and
                  K. Cho and
                  A. Oh},
  title        = {Transformer Memory as a Differentiable Search Index},
  booktitle    = {Advances in Neural Information Processing Systems 35: Annual Conference
                  on Neural Information Processing Systems 2022, NeurIPS 2022, New Orleans,
                  LA, USA, November 28 - December 9, 2022},
  year         = {2022},
  url          = {http://papers.nips.cc/paper\_files/paper/2022/hash/892840a6123b5ec99ebaab8be1530fba-Abstract-Conference.html},
  bibsource    = {dblp computer science bibliography, https://dblp.org}
}

@inproceedings{chen2023understanding,
  title={Understanding differential search index for text retrieval},
  author={Chen, Xiaoyang and Liu, Yanjiang and He, Ben and Sun, Le and Sun, Yingfei},
  booktitle={Findings of the association for computational linguistics: ACL 2023},
  pages={10701--10717},
  year={2023}
}

@inproceedings{3295222.3295349,
author = {Vaswani, Ashish and Shazeer, Noam and Parmar, Niki and Uszkoreit, Jakob and Jones, Llion and Gomez, Aidan N. and Kaiser, \L{}ukasz and Polosukhin, Illia},
title = {Attention is all you need},
year = {2017},
isbn = {9781510860964},
publisher = {Curran Associates Inc.},
address = {Red Hook, NY, USA},
booktitle = {Proceedings of the 31st International Conference on Neural Information Processing Systems},
pages = {6000–6010},
numpages = {11},
location = {Long Beach, California, USA},
series = {NIPS'17}
}

@inproceedings{DBLP:conf/nips/RajputMSKVHHT0S23,
  author       = {Shashank Rajput and
                  Nikhil Mehta and
                  Anima Singh and
                  Raghunandan Hulikal Keshavan and
                  Trung Vu and
                  Lukasz Heldt and
                  Lichan Hong and
                  Yi Tay and
                  Vinh Q. Tran and
                  Jonah Samost and
                  Maciej Kula and
                  Ed H. Chi and
                  Mahesh Sathiamoorthy},
  editor       = {Alice Oh and
                  Tristan Naumann and
                  Amir Globerson and
                  Kate Saenko and
                  Moritz Hardt and
                  Sergey Levine},
  title        = {Recommender Systems with Generative Retrieval},
  booktitle    = {Advances in Neural Information Processing Systems 36: Annual Conference
                  on Neural Information Processing Systems 2023, NeurIPS 2023, New Orleans,
                  LA, USA, December 10 - 16, 2023},
  year         = {2023},
  url          = {http://papers.nips.cc/paper\_files/paper/2023/hash/20dcab0f14046a5c6b02b61da9f13229-Abstract-Conference.html},
  bibsource    = {dblp computer science bibliography, https://dblp.org}
}

@article{DBLP:journals/corr/abs-2512-24762,
  author       = {OneRec Team},
  title        = {OpenOneRec Technical Report},
  journal      = {CoRR},
  volume       = {abs/2512.24762},
  year         = {2025},
  url          = {https://doi.org/10.48550/arXiv.2512.24762},
  doi          = {10.48550/ARXIV.2512.24762},
  eprinttype   = {arXiv},
  eprint       = {2512.24762},
  bibsource    = {dblp computer science bibliography, https://dblp.org}
}

@article{DBLP:journals/corr/abs-2402-03300,
  author       = {Zhihong Shao and
                  Peiyi Wang and
                  Qihao Zhu and
                  Runxin Xu and
                  Junxiao Song and
                  Mingchuan Zhang and
                  Y. K. Li and
                  Y. Wu and
                  Daya Guo},
  title        = {DeepSeekMath: Pushing the Limits of Mathematical Reasoning in Open
                  Language Models},
  journal      = {CoRR},
  volume       = {abs/2402.03300},
  year         = {2024},
  url          = {https://doi.org/10.48550/arXiv.2402.03300},
  doi          = {10.48550/ARXIV.2402.03300},
  eprinttype   = {arXiv},
  eprint       = {2402.03300},
  bibsource    = {dblp computer science bibliography, https://dblp.org}
}

@inproceedings{DBLP:conf/sigir/AiBGC18,
  author       = {Qingyao Ai and
                  Keping Bi and
                  Jiafeng Guo and
                  W. Bruce Croft},
  editor       = {Kevyn Collins{-}Thompson and
                  Qiaozhu Mei and
                  Brian D. Davison and
                  Yiqun Liu and
                  Emine Yilmaz},
  title        = {Learning a Deep Listwise Context Model for Ranking Refinement},
  booktitle    = {The 41st International {ACM} {SIGIR} Conference on Research {\&}
                  Development in Information Retrieval, {SIGIR} 2018, Ann Arbor, MI,
                  USA, July 08-12, 2018},
  pages        = {135--144},
  publisher    = {{ACM}},
  year         = {2018},
  url          = {https://doi.org/10.1145/3209978.3209985},
  doi          = {10.1145/3209978.3209985},
  bibsource    = {dblp computer science bibliography, https://dblp.org}
}

@inproceedings{DBLP:conf/sigir/PangXALCW20,
  author       = {Liang Pang and
                  Jun Xu and
                  Qingyao Ai and
                  Yanyan Lan and
                  Xueqi Cheng and
                  Jirong Wen},
  editor       = {Jimmy X. Huang and
                  Yi Chang and
                  Xueqi Cheng and
                  Jaap Kamps and
                  Vanessa Murdock and
                  Ji{-}Rong Wen and
                  Yiqun Liu},
  title        = {SetRank: Learning a Permutation-Invariant Ranking Model for Information
                  Retrieval},
  booktitle    = {Proceedings of the 43rd International {ACM} {SIGIR} conference on
                  research and development in Information Retrieval, {SIGIR} 2020, Virtual
                  Event, China, July 25-30, 2020},
  pages        = {499--508},
  publisher    = {{ACM}},
  year         = {2020},
  url          = {https://doi.org/10.1145/3397271.3401104},
  doi          = {10.1145/3397271.3401104},
  bibsource    = {dblp computer science bibliography, https://dblp.org}
}

@inproceedings{DBLP:conf/www/LiZLSCZTXH22,
  author       = {Yi Li and
                  Jieming Zhu and
                  Weiwen Liu and
                  Liangcai Su and
                  Guohao Cai and
                  Qi Zhang and
                  Ruiming Tang and
                  Xi Xiao and
                  Xiuqiang He},
  editor       = {Fr{\'{e}}d{\'{e}}rique Laforest and
                  Rapha{\"{e}}l Troncy and
                  Elena Simperl and
                  Deepak Agarwal and
                  Aristides Gionis and
                  Ivan Herman and
                  Lionel M{\'{e}}dini},
  title        = {{PEAR:} Personalized Re-ranking with Contextualized Transformer for
                  Recommendation},
  booktitle    = {Companion of The Web Conference 2022, Virtual Event / Lyon, France,
                  April 25 - 29, 2022},
  pages        = {62--66},
  publisher    = {{ACM}},
  year         = {2022},
  url          = {https://doi.org/10.1145/3487553.3524208},
  doi          = {10.1145/3487553.3524208},
  bibsource    = {dblp computer science bibliography, https://dblp.org}
}

@inproceedings{DBLP:conf/sigir/XiLZZDTZZY22,
  author       = {Yunjia Xi and
                  Weiwen Liu and
                  Jieming Zhu and
                  Xilong Zhao and
                  Xinyi Dai and
                  Ruiming Tang and
                  Weinan Zhang and
                  Rui Zhang and
                  Yong Yu},
  editor       = {Enrique Amig{\'{o}} and
                  Pablo Castells and
                  Julio Gonzalo and
                  Ben Carterette and
                  J. Shane Culpepper and
                  Gabriella Kazai},
  title        = {Multi-Level Interaction Reranking with User Behavior History},
  booktitle    = {{SIGIR} '22: The 45th International {ACM} {SIGIR} Conference on Research
                  and Development in Information Retrieval, Madrid, Spain, July 11 -
                  15, 2022},
  pages        = {1336--1346},
  publisher    = {{ACM}},
  year         = {2022},
  url          = {https://doi.org/10.1145/3477495.3532026},
  doi          = {10.1145/3477495.3532026},
  bibsource    = {dblp computer science bibliography, https://dblp.org}
}

@article{DBLP:journals/corr/abs-1810-02019,
  author       = {Irwan Bello and
                  Sayali Kulkarni and
                  Sagar Jain and
                  Craig Boutilier and
                  Ed Huai{-}hsin Chi and
                  Elad Eban and
                  Xiyang Luo and
                  Alan Mackey and
                  Ofer Meshi},
  title        = {Seq2Slate: Re-ranking and Slate Optimization with RNNs},
  journal      = {CoRR},
  volume       = {abs/1810.02019},
  year         = {2018},
  url          = {http://arxiv.org/abs/1810.02019},
  eprinttype   = {arXiv},
  eprint       = {1810.02019},
  bibsource    = {dblp computer science bibliography, https://dblp.org}
}

@inproceedings{DBLP:conf/ijcai/ZhuangOW18,
  author       = {Tao Zhuang and
                  Wenwu Ou and
                  Zhirong Wang},
  editor       = {J{\'{e}}r{\^{o}}me Lang},
  title        = {Globally Optimized Mutual Influence Aware Ranking in E-Commerce Search},
  booktitle    = {Proceedings of the Twenty-Seventh International Joint Conference on
                  Artificial Intelligence, {IJCAI} 2018, July 13-19, 2018, Stockholm,
                  Sweden},
  pages        = {3725--3731},
  publisher    = {ijcai.org},
  year         = {2018},
  url          = {https://doi.org/10.24963/ijcai.2018/518},
  doi          = {10.24963/IJCAI.2018/518},
  bibsource    = {dblp computer science bibliography, https://dblp.org}
}

@inproceedings{DBLP:conf/iclr/JiangGQMR19,
  author       = {Ray Jiang and
                  Sven Gowal and
                  Yuqiu Qian and
                  Timothy A. Mann and
                  Danilo J. Rezende},
  title        = {Beyond Greedy Ranking: Slate Optimization via List-CVAE},
  booktitle    = {7th International Conference on Learning Representations, {ICLR} 2019,
                  New Orleans, LA, USA, May 6-9, 2019},
  publisher    = {OpenReview.net},
  year         = {2019},
  url          = {https://openreview.net/forum?id=r1xX42R5Fm},
  bibsource    = {dblp computer science bibliography, https://dblp.org}
}

@inproceedings{DBLP:conf/kdd/ShiYWWGLWWW23,
  author       = {Xiaowen Shi and
                  Fan Yang and
                  Ze Wang and
                  Xiaoxu Wu and
                  Muzhi Guan and
                  Guogang Liao and
                  Yongkang Wang and
                  Xingxing Wang and
                  Dong Wang},
  editor       = {Ambuj K. Singh and
                  Yizhou Sun and
                  Leman Akoglu and
                  Dimitrios Gunopulos and
                  Xifeng Yan and
                  Ravi Kumar and
                  Fatma Ozcan and
                  Jieping Ye},
  title        = {{PIER:} Permutation-Level Interest-Based End-to-End Re-ranking Framework
                  in E-commerce},
  booktitle    = {Proceedings of the 29th {ACM} {SIGKDD} Conference on Knowledge Discovery
                  and Data Mining, {KDD} 2023, Long Beach, CA, USA, August 6-10, 2023},
  pages        = {4823--4831},
  publisher    = {{ACM}},
  year         = {2023},
  url          = {https://doi.org/10.1145/3580305.3599886},
  doi          = {10.1145/3580305.3599886},
  bibsource    = {dblp computer science bibliography, https://dblp.org}
}

@article{DBLP:journals/corr/abs-2005-12206,
  author       = {Jianxiong Wei and
                  Anxiang Zeng and
                  Yueqiu Wu and
                  Peng Guo and
                  Qingsong Hua and
                  Qingpeng Cai},
  title        = {Generator and Critic: {A} Deep Reinforcement Learning Approach for
                  Slate Re-ranking in E-commerce},
  journal      = {CoRR},
  volume       = {abs/2005.12206},
  year         = {2020},
  url          = {https://arxiv.org/abs/2005.12206},
  eprinttype   = {arXiv},
  eprint       = {2005.12206},
  bibsource    = {dblp computer science bibliography, https://dblp.org}
}

@article{DBLP:journals/corr/abs-2102-12057,
  author       = {Yufei Feng and
                  Yu Gong and
                  Fei Sun and
                  Qingwen Liu and
                  Wenwu Ou},
  title        = {Revisit Recommender System in the Permutation Prospective},
  journal      = {CoRR},
  volume       = {abs/2102.12057},
  year         = {2021},
  url          = {https://arxiv.org/abs/2102.12057},
  eprinttype   = {arXiv},
  eprint       = {2102.12057},
  bibsource    = {dblp computer science bibliography, https://dblp.org}
}

@inproceedings{DBLP:conf/www/WangWKWCTZXW25,
  author       = {Shuli Wang and
                  Xue Wei and
                  Senjie Kou and
                  Chi Wang and
                  Wenshuai Chen and
                  Qi Tang and
                  Yinhua Zhu and
                  Xiong Xiao and
                  Xingxing Wang},
  editor       = {Guodong Long and
                  Michale Blumestein and
                  Yi Chang and
                  Liane Lewin{-}Eytan and
                  Zi Helen Huang and
                  Elad Yom{-}Tov},
  title        = {{NLGR:} Utilizing Neighbor Lists for Generative Rerank in Personalized
                  Recommendation Systems},
  booktitle    = {Companion Proceedings of the {ACM} on Web Conference 2025, {WWW} 2025,
                  Sydney, NSW, Australia, 28 April 2025 - 2 May 2025},
  pages        = {530--537},
  publisher    = {{ACM}},
  year         = {2025},
  url          = {https://doi.org/10.1145/3701716.3715251},
  doi          = {10.1145/3701716.3715251},
  bibsource    = {dblp computer science bibliography, https://dblp.org}
}

@inproceedings{10.1145/3637528.3671645,
author = {Ren, Yuxin and Yang, Qiya and Wu, Yichun and Xu, Wei and Wang, Yalong and Zhang, Zhiqiang},
title = {Non-autoregressive Generative Models for Reranking Recommendation},
year = {2024},
isbn = {9798400704901},
publisher = {Association for Computing Machinery},
address = {New York, NY, USA},
url = {https://doi.org/10.1145/3637528.3671645},
doi = {10.1145/3637528.3671645},
booktitle = {Proceedings of the 30th ACM SIGKDD Conference on Knowledge Discovery and Data Mining},
pages = {5625–5634},
numpages = {10},
location = {Barcelona, Spain},
series = {KDD '24}
}

@inproceedings{DBLP:conf/kdd/00060HSMZ0G23,
  author       = {Shuchang Liu and
                  Qingpeng Cai and
                  Zhankui He and
                  Bowen Sun and
                  Julian J. McAuley and
                  Dong Zheng and
                  Peng Jiang and
                  Kun Gai},
  editor       = {Ambuj K. Singh and
                  Yizhou Sun and
                  Leman Akoglu and
                  Dimitrios Gunopulos and
                  Xifeng Yan and
                  Ravi Kumar and
                  Fatma Ozcan and
                  Jieping Ye},
  title        = {Generative Flow Network for Listwise Recommendation},
  booktitle    = {Proceedings of the 29th {ACM} {SIGKDD} Conference on Knowledge Discovery
                  and Data Mining, {KDD} 2023, Long Beach, CA, USA, August 6-10, 2023},
  pages        = {1524--1534},
  publisher    = {{ACM}},
  year         = {2023},
  url          = {https://doi.org/10.1145/3580305.3599364},
  doi          = {10.1145/3580305.3599364},
  bibsource    = {dblp computer science bibliography, https://dblp.org}
}

@article{DBLP:journals/corr/abs-2104-00860,
  author       = {Yufei Feng and
                  Binbin Hu and
                  Yu Gong and
                  Fei Sun and
                  Qingwen Liu and
                  Wenwu Ou},
  title        = {{GRN:} Generative Rerank Network for Context-wise Recommendation},
  journal      = {CoRR},
  volume       = {abs/2104.00860},
  year         = {2021},
  url          = {https://arxiv.org/abs/2104.00860},
  eprinttype   = {arXiv},
  eprint       = {2104.00860},
  bibsource    = {dblp computer science bibliography, https://dblp.org}
}

@inproceedings{DBLP:conf/www/XuXCLSLWZZ26,
  author       = {Junwei Xu and
                  Zhibo Xiao and
                  Chuxin Chen and
                  Chengyu Lai and
                  Qijie Shen and
                  Jiuning Lin and
                  Dimin Wang and
                  Jialin Zhu and
                  Xiao{-}Ping Zhang},
  editor       = {Hakim Hacid and
                  Yoelle Maarek and
                  Francesco Bonchi and
                  Ido Guy and
                  Emine Yilmaz},
  title        = {OMGRec: One-time Matching-based Generative Rerank with Permutation-level
                  Modeling in E-commerce},
  booktitle    = {Proceedings of the {ACM} Web Conference 2026, {WWW} 2026, Dubai, United
                  Arab Emirates, originally scheduled for April 13-17, 2026, rescheduled
                  for June 29 - July 3, 2026},
  pages        = {8421--8424},
  publisher    = {{ACM}},
  year         = {2026},
  url          = {https://doi.org/10.1145/3774904.3792875},
  doi          = {10.1145/3774904.3792875},
  bibsource    = {dblp computer science bibliography, https://dblp.org}
}

@inproceedings{DBLP:conf/atal/ZhaoLFGZHCPD24,
  author       = {Xin Zhao and
                  Jiaxin Li and
                  Zhiwei Fang and
                  Yuchen Guo and
                  Jinyuan Zhao and
                  Jie He and
                  Wenlong Chen and
                  Changping Peng and
                  Guiguang Ding},
  editor       = {Mehdi Dastani and
                  Jaime Sim{\~{a}}o Sichman and
                  Natasha Alechina and
                  Virginia Dignum},
  title        = {JDRec: Practical Actor-Critic Framework for Online Combinatorial Recommender
                  System},
  booktitle    = {Proceedings of the 23rd International Conference on Autonomous Agents
                  and Multiagent Systems, {AAMAS} 2024, Auckland, New Zealand, May 6-10,
                  2024},
  pages        = {2612--2614},
  publisher    = {International Foundation for Autonomous Agents and Multiagent Systems
                  / {ACM}},
  year         = {2024},
  url          = {https://dl.acm.org/doi/10.5555/3635637.3663244},
  doi          = {10.5555/3635637.3663244},
  bibsource    = {dblp computer science bibliography, https://dblp.org}
}

@inproceedings{DBLP:conf/kdd/ZhuLCMSZZX0025,
  author       = {Ruitao Zhu and
                  Yangsu Liu and
                  Dagui Chen and
                  Zhenjia Ma and
                  Chufeng Shi and
                  Zhenzhe Zheng and
                  Jie Zhang and
                  Jian Xu and
                  Bo Zheng and
                  Fan Wu},
  editor       = {Yizhou Sun and
                  Flavio Chierichetti and
                  Hady W. Lauw and
                  Claudia Perlich and
                  Wee Hyong Tok and
                  Andrew Tomkins},
  title        = {Contextual Generative Auction with Permutation-level Externalities
                  for Online Advertising},
  booktitle    = {Proceedings of the 31st {ACM} {SIGKDD} Conference on Knowledge Discovery
                  and Data Mining, V.1, {KDD} 2025, Toronto, ON, Canada, August 3-7,
                  2025},
  pages        = {2171--2181},
  publisher    = {{ACM}},
  year         = {2025},
  url          = {https://doi.org/10.1145/3690624.3709313},
  doi          = {10.1145/3690624.3709313},
  bibsource    = {dblp computer science bibliography, https://dblp.org}
}

@inproceedings{DBLP:conf/sigir/YangQ0YWLHLG25,
  author       = {Hailan Yang and
                  Zhenyu Qi and
                  Shuchang Liu and
                  Xiaoyu Yang and
                  Xiaobei Wang and
                  Xiang Li and
                  Lantao Hu and
                  Han Li and
                  Kun Gai},
  editor       = {Nicola Ferro and
                  Maria Maistro and
                  Gabriella Pasi and
                  Omar Alonso and
                  Andrew Trotman and
                  Suzan Verberne},
  title        = {Comprehensive List Generation for Multi-Generator Reranking},
  booktitle    = {Proceedings of the 48th International {ACM} {SIGIR} Conference on
                  Research and Development in Information Retrieval, {SIGIR} 2025, Padua,
                  Italy, July 13-18, 2025},
  pages        = {2298--2308},
  publisher    = {{ACM}},
  year         = {2025},
  url          = {https://doi.org/10.1145/3726302.3729933},
  doi          = {10.1145/3726302.3729933},
  bibsource    = {dblp computer science bibliography, https://dblp.org}
}

@inproceedings{10.1145/3746252.3761539,
author = {Wang, Shuli and Huang, Yinqiu and Li, Changhao and Zhou, Yuan and Liu, Yonggang and Zhang, Yongqiang and Zhu, Yinhua and Wang, Haitao and Wang, Xingxing},
title = {You Only Evaluate Once: A Tree-based Rerank Method at Meituan},
year = {2025},
isbn = {9798400720406},
publisher = {Association for Computing Machinery},
address = {New York, NY, USA},
url = {https://doi.org/10.1145/3746252.3761539},
doi = {10.1145/3746252.3761539},
booktitle = {Proceedings of the 34th ACM International Conference on Information and Knowledge Management},
pages = {6136–6143},
numpages = {8},
location = {Seoul, Republic of Korea},
series = {CIKM '25}
}

@inproceedings{zhang2026goalrank,
title={GoalRank: Group-Relative Optimization for a Large Ranking Model},
author={Kaike Zhang and Xiaobei Wang and Shuchang Liu and Hailan Yang and Xiang Li and Lantao Hu and Han Li and Qi Cao and Fei Sun and Kun Gai},
booktitle={The Fourteenth International Conference on Learning Representations},
year={2026},
url={https://openreview.net/forum?id=gTMzRm8fb0}
}

@inproceedings{2969239.2969370,
author = {Bengio, Samy and Vinyals, Oriol and Jaitly, Navdeep and Shazeer, Noam},
title = {Scheduled sampling for sequence prediction with recurrent Neural networks},
year = {2015},
publisher = {MIT Press},
address = {Cambridge, MA, USA},
booktitle = {Proceedings of the 29th International Conference on Neural Information Processing Systems - Volume 1},
pages = {1171–1179},
numpages = {9},
location = {Montreal, Canada},
series = {NIPS'15}
}

@inproceedings{DBLP:conf/nips/GoyalLZZCB16,
  author       = {Anirudh Goyal and
                  Alex Lamb and
                  Ying Zhang and
                  Saizheng Zhang and
                  Aaron C. Courville and
                  Yoshua Bengio},
  editor       = {Daniel D. Lee and
                  Masashi Sugiyama and
                  Ulrike von Luxburg and
                  Isabelle Guyon and
                  Roman Garnett},
  title        = {Professor Forcing: {A} New Algorithm for Training Recurrent Networks},
  booktitle    = {Advances in Neural Information Processing Systems 29: Annual Conference
                  on Neural Information Processing Systems 2016, December 5-10, 2016,
                  Barcelona, Spain},
  pages        = {4601--4609},
  year         = {2016},
  url          = {https://proceedings.neurips.cc/paper/2016/hash/16026d60ff9b54410b3435b403afd226-Abstract.html},
  bibsource    = {dblp computer science bibliography, https://dblp.org}
}

@inproceedings{shen-etal-2016-minimum,
    title = "Minimum Risk Training for Neural Machine Translation",
    author = "Shen, Shiqi  and
      Cheng, Yong  and
      He, Zhongjun  and
      He, Wei  and
      Wu, Hua  and
      Sun, Maosong  and
      Liu, Yang",
    editor = "Erk, Katrin  and
      Smith, Noah A.",
    booktitle = "Proceedings of the 54th Annual Meeting of the Association for Computational Linguistics (Volume 1: Long Papers)",
    month = aug,
    year = "2016",
    address = "Berlin, Germany",
    publisher = "Association for Computational Linguistics",
    url = "https://aclanthology.org/P16-1159/",
    doi = "10.18653/v1/P16-1159",
    pages = "1683--1692"
}

@inproceedings{3298483.3298649,
author = {Yu, Lantao and Zhang, Weinan and Wang, Jun and Yu, Yong},
title = {SeqGAN: sequence generative adversarial nets with policy gradient},
year = {2017},
publisher = {AAAI Press},
booktitle = {Proceedings of the Thirty-First AAAI Conference on Artificial Intelligence},
pages = {2852–2858},
numpages = {7},
location = {San Francisco, California, USA},
series = {AAAI'17}
}

@inproceedings{schmidt-etal-2022-non,
    title = "Non-Autoregressive Neural Machine Translation: A Call for Clarity",
    author = {Schmidt, Robin  and
      Pires, Telmo  and
      Peitz, Stephan  and
      L{\"o}{\"o}f, Jonas},
    editor = "Goldberg, Yoav  and
      Kozareva, Zornitsa  and
      Zhang, Yue",
    booktitle = "Proceedings of the 2022 Conference on Empirical Methods in Natural Language Processing",
    month = dec,
    year = "2022",
    address = "Abu Dhabi, United Arab Emirates",
    publisher = "Association for Computational Linguistics",
    url = "https://aclanthology.org/2022.emnlp-main.179/",
    doi = "10.18653/v1/2022.emnlp-main.179",
    pages = "2785--2799"
}

@inproceedings{DBLP:conf/nips/HuangLHZS25,
  author       = {Xun Huang and
                  Zhengqi Li and
                  Guande He and
                  Mingyuan Zhou and
                  Eli Shechtman},
  editor       = {Danielle Belgrave and
                  Cheng Zhang and
                  Laura N. Montoya and
                  Hsuan{-}Tien Lin and
                  Razvan Pascanu and
                  Piotr Koniusz and
                  Marzyeh Ghassemi and
                  Nancy Chen and
                  Iv{\'{a}}n Vladimir Meza Ru{\'{\i}}z and
                  Arturo Loaiza{-}Bonilla},
  title        = {Self Forcing: Bridging the Train-Test Gap in Autoregressive Video
                  Diffusion},
  booktitle    = {Advances in Neural Information Processing Systems 38: Annual Conference
                  on Neural Information Processing Systems 2025, NeurIPS 2025, San Diego,
                  CA, USA, December 2-7, 2025 / Mexico City, Mexico, November 30 - December
                  5, 2025},
  year         = {2025},
  url          = {http://papers.nips.cc/paper\_files/paper/2025/hash/f4823f831af67a3ef15e41a85434422a-Abstract-Conference.html},
  bibsource    = {dblp computer science bibliography, https://dblp.org}
}

@inproceedings{DBLP:conf/cvpr/ChangZJLF22,
  author       = {Huiwen Chang and
                  Han Zhang and
                  Lu Jiang and
                  Ce Liu and
                  William T. Freeman},
  title        = {MaskGIT: Masked Generative Image Transformer},
  booktitle    = {{IEEE/CVF} Conference on Computer Vision and Pattern Recognition,
                  {CVPR} 2022, New Orleans, LA, USA, June 18-24, 2022},
  pages        = {11305--11315},
  publisher    = {{IEEE}},
  year         = {2022},
  url          = {https://doi.org/10.1109/CVPR52688.2022.01103},
  doi          = {10.1109/CVPR52688.2022.01103},
  bibsource    = {dblp computer science bibliography, https://dblp.org}
}

@article{guo2026promise,
  title={PROMISE: Process Reward Models Unlock Test-Time Scaling Laws in Generative Recommendations},
  author={Guo, Chengcheng and Cai, Kuo and Zhou, Yu and Luo, Qiang and Tang, Ruiming and Li, Han and Gai, Kun and Zhou, Guorui},
  journal={arXiv preprint arXiv:2601.04674},
  year={2026}
}

@inproceedings{DBLP:conf/iclr/LightmanKBEBLLS24,
  author       = {Hunter Lightman and
                  Vineet Kosaraju and
                  Yuri Burda and
                  Harrison Edwards and
                  Bowen Baker and
                  Teddy Lee and
                  Jan Leike and
                  John Schulman and
                  Ilya Sutskever and
                  Karl Cobbe},
  title        = {Let's Verify Step by Step},
  booktitle    = {The Twelfth International Conference on Learning Representations,
                  {ICLR} 2024, Vienna, Austria, May 7-11, 2024},
  publisher    = {OpenReview.net},
  year         = {2024},
  url          = {https://openreview.net/forum?id=v8L0pN6EOi},
  bibsource    = {dblp computer science bibliography, https://dblp.org}
}

@article{10.1145/1498765.1498785,
author = {Williams, Samuel and Waterman, Andrew and Patterson, David},
title = {Roofline: an insightful visual performance model for multicore architectures},
year = {2009},
issue_date = {April 2009},
publisher = {Association for Computing Machinery},
address = {New York, NY, USA},
volume = {52},
number = {4},
issn = {0001-0782},
url = {https://doi.org/10.1145/1498765.1498785},
doi = {10.1145/1498765.1498785},
journal = {Commun. ACM},
month = apr,
pages = {65–76},
numpages = {12}
}

@inproceedings{DBLP:conf/sigir/MengGCLZ25,
  author       = {Yue Meng and
                  Cheng Guo and
                  Yi Cao and
                  Tong Liu and
                  Bo Zheng},
  editor       = {Nicola Ferro and
                  Maria Maistro and
                  Gabriella Pasi and
                  Omar Alonso and
                  Andrew Trotman and
                  Suzan Verberne},
  title        = {A Generative Re-ranking Model for List-level Multi-objective Optimization
                  at Taobao},
  booktitle    = {Proceedings of the 48th International {ACM} {SIGIR} Conference on
                  Research and Development in Information Retrieval, {SIGIR} 2025, Padua,
                  Italy, July 13-18, 2025},
  pages        = {4213--4218},
  publisher    = {{ACM}},
  year         = {2025},
  url          = {https://doi.org/10.1145/3726302.3731935},
  doi          = {10.1145/3726302.3731935},
  bibsource    = {dblp computer science bibliography, https://dblp.org}
}

@inproceedings{645531.656005,
author = {Kakade, Sham and Langford, John},
title = {Approximately Optimal Approximate Reinforcement Learning},
year = {2002},
isbn = {1558608737},
publisher = {Morgan Kaufmann Publishers Inc.},
address = {San Francisco, CA, USA},
booktitle = {Proceedings of the Nineteenth International Conference on Machine Learning},
pages = {267–274},
numpages = {8},
series = {ICML '02}
}

@InProceedings{pmlr-v9-ross10a,
  title = 	 {Efficient Reductions for Imitation Learning},
  author = 	 {Ross, Stephane and Bagnell, Drew},
  booktitle = 	 {Proceedings of the Thirteenth International Conference on Artificial Intelligence and Statistics},
  pages = 	 {661--668},
  year = 	 {2010},
  editor = 	 {Teh, Yee Whye and Titterington, Mike},
  volume = 	 {9},
  series = 	 {Proceedings of Machine Learning Research},
  address = 	 {Chia Laguna Resort, Sardinia, Italy},
  month = 	 {13--15 May},
  publisher =    {PMLR},
  url = 	 {https://proceedings.mlr.press/v9/ross10a.html}
}

@article{harper2016movielens,
  author       = {F. Maxwell Harper and
                  Joseph A. Konstan},
  title        = {The MovieLens Datasets: History and Context},
  journal      = {{ACM} Trans. Interact. Intell. Syst.},
  volume       = {5},
  number       = {4},
  pages        = {19:1--19:19},
  year         = {2016},
  url          = {https://doi.org/10.1145/2827872},
  doi          = {10.1145/2827872},
  bibsource    = {dblp computer science bibliography, https://dblp.org}
}

@inproceedings{he2016ups,
  author       = {Ruining He and
                  Julian J. McAuley},
  editor       = {Jacqueline Bourdeau and
                  Jim Hendler and
                  Roger Nkambou and
                  Ian Horrocks and
                  Ben Y. Zhao},
  title        = {Ups and Downs: Modeling the Visual Evolution of Fashion Trends with
                  One-Class Collaborative Filtering},
  booktitle    = {Proceedings of the 25th International Conference on World Wide Web,
                  {WWW} 2016, Montreal, Canada, April 11 - 15, 2016},
  pages        = {507--517},
  publisher    = {{ACM}},
  year         = {2016},
  url          = {https://doi.org/10.1145/2872427.2883037},
  doi          = {10.1145/2872427.2883037},
  bibsource    = {dblp computer science bibliography, https://dblp.org}
}

@inproceedings{liu2024recflow,
  author       = {Qi Liu and
                  Kai Zheng and
                  Rui Huang and
                  Wuchao Li and
                  Kuo Cai and
                  Yuan Chai and
                  Yanan Niu and
                  Yiqun Hui and
                  Bing Han and
                  Na Mou and
                  Hongning Wang and
                  Wentian Bao and
                  Yunen Yu and
                  Guorui Zhou and
                  Han Li and
                  Yang Song and
                  Defu Lian and
                  Kun Gai},
  title        = {RecFlow: An Industrial Full Flow Recommendation Dataset},
  booktitle    = {The Thirteenth International Conference on Learning Representations,
                  {ICLR} 2025, Singapore, April 24-28, 2025},
  publisher    = {OpenReview.net},
  year         = {2025},
  url          = {https://openreview.net/forum?id=vVHc8bGRns},
  bibsource    = {dblp computer science bibliography, https://dblp.org}
}

@inproceedings{DBLP:conf/recsys/CovingtonAS16,
  author       = {Paul Covington and
                  Jay Adams and
                  Emre Sargin},
  editor       = {Shilad Sen and
                  Werner Geyer and
                  Jill Freyne and
                  Pablo Castells},
  title        = {Deep Neural Networks for YouTube Recommendations},
  booktitle    = {Proceedings of the 10th {ACM} Conference on Recommender Systems, Boston,
                  MA, USA, September 15-19, 2016},
  pages        = {191--198},
  publisher    = {{ACM}},
  year         = {2016},
  url          = {https://doi.org/10.1145/2959100.2959190},
  doi          = {10.1145/2959100.2959190},
  bibsource    = {dblp computer science bibliography, https://dblp.org}
}

@inproceedings{DBLP:conf/kdd/WangFFW17,
  author       = {Ruoxi Wang and
                  Bin Fu and
                  Gang Fu and
                  Mingliang Wang},
  title        = {Deep {\&} Cross Network for Ad Click Predictions},
  booktitle    = {Proceedings of the ADKDD'17, Halifax, NS, Canada, August 13 - 17,
                  2017},
  pages        = {12:1--12:7},
  publisher    = {{ACM}},
  year         = {2017},
  url          = {https://doi.org/10.1145/3124749.3124754},
  doi          = {10.1145/3124749.3124754},
  bibsource    = {dblp computer science bibliography, https://dblp.org}
}

@inproceedings{DBLP:conf/recsys/PeiZZSLSWJGOP19,
  author       = {Changhua Pei and
                  Yi Zhang and
                  Yongfeng Zhang and
                  Fei Sun and
                  Xiao Lin and
                  Hanxiao Sun and
                  Jian Wu and
                  Peng Jiang and
                  Junfeng Ge and
                  Wenwu Ou and
                  Dan Pei},
  editor       = {Toine Bogers and
                  Alan Said and
                  Peter Brusilovsky and
                  Domonkos Tikk},
  title        = {Personalized re-ranking for recommendation},
  booktitle    = {Proceedings of the 13th {ACM} Conference on Recommender Systems, RecSys
                  2019, Copenhagen, Denmark, September 16-20, 2019},
  pages        = {3--11},
  publisher    = {{ACM}},
  year         = {2019},
  url          = {https://doi.org/10.1145/3298689.3347000},
  doi          = {10.1145/3298689.3347000},
  bibsource    = {dblp computer science bibliography, https://dblp.org}
}

\appendix
\section{Odd L Extension}
\label{appendix:odd_L}

The main paper assumes $L$ is even for notational simplicity. Here we show that all theoretical guarantees extend to odd $L$ with negligible degradation.

\subsection{Generation Protocol}

When $L$ is odd, the pair-space generator produces $\lfloor L/2 \rfloor$ pair tokens followed by a single item token at the final step:
\[
S = \left\lfloor \frac{L}{2} \right\rfloor + 1 = \frac{L+1}{2}.
\]
At the final step, the output projection switches from the $n^2$-sized pair vocabulary to the remaining $n - (L-1)$ items, costing $O(nd)$ --- strictly less than the $O(n^2 d)$ of a pair step. The decoder's self-attention and cross-attention mechanisms remain unchanged.

\subsection{Generation Complexity (Theorem \ref{thm:complexity} Extension)}

\begin{proposition}[Odd-$L$ complexity]
\label{prop:odd-complexity}
Under the same setting as Theorem~\ref{thm:complexity}, for odd $L$ the per-request FLOPs satisfy:
\begin{align*}
T_{\mathrm{fixed}}^{(\mathrm{pair})} &= O\!\big(H(d^2 + dd_{\mathrm{ff}} + bd)(L+1)/2\big), &
\text{Ratio: } &\frac{2L}{L+1}\times\;\downarrow \\
T_{\mathrm{kv}}^{(\mathrm{pair})} &= O\!\big(Hd((L+1)/2)^2\big), &
\text{Ratio: } &\frac{4L^2}{(L+1)^2}\times\;\downarrow \\
T_{\mathrm{proj}}^{(\mathrm{pair})} &= O\!\big(n^2 d \lfloor L/2\rfloor + nd\big), &
\text{Ratio: } &\approx \frac{n}{2}\times\;\uparrow
\end{align*}
where the ratios are relative to item-space generation at horizon $S = L$.
\end{proposition}

\begin{proof}
The decoder executes $\lfloor L/2 \rfloor = (L-1)/2$ pair steps and one item step.

\textbf{Fixed cost + cross-attention.} Each of the $S = (L+1)/2$ steps incurs the same per-step fixed cost $O(H(d^2 + dd_{\mathrm{ff}} + bd))$. Total:
\[
T_{\mathrm{fixed}}^{(\mathrm{pair})} = O\!\big(H(d^2 + dd_{\mathrm{ff}} + bd)(L+1)/2\big).
\]
The ratio to item-space is $L / ((L+1)/2) = 2L/(L+1)$.

\textbf{KV-cache read.} The KV cache grows from length $1$ to $S = (L+1)/2$ over the $S$ decoding steps. Total:
\[
T_{\mathrm{kv}}^{(\mathrm{pair})} = O\!\big(Hd \cdot S(S+1)/2\big) = O\!\big(Hd((L+1)/2)^2\big).
\]
The ratio to item-space is $L^2 / ((L+1)/2)^2 = 4L^2/(L+1)^2$.

\textbf{Output projection.} The first $(L-1)/2$ steps project over $n^2$ pair tokens, costing $O(n^2 d \cdot (L-1)/2)$. The final step projects over $n - (L-1)$ remaining items, costing $O(nd)$. Total:
\[
T_{\mathrm{proj}}^{(\mathrm{pair})} = O(n^2 d (L-1)/2 + nd) = O(n^2 d \lfloor L/2 \rfloor + nd).
\]
Since $nd \ll n^2 d \lfloor L/2 \rfloor$ for $n \geq 30$ and $L \geq 3$, this is $\approx O(n^2 d L / 2)$, same as the even-$L$ case. The ratio to item-space is $\approx n/2$.
\end{proof}

The deviation from the even-$L$ ratios is:
\begin{table}[h]
\centering
\caption{Speedup ratios for even vs.\ odd $L$.}
\footnotesize
\setlength{\tabcolsep}{3pt}
\begin{tabular}{cccc}
\toprule
$L$ & Fixed-cost $2L/(L{+}1)$ & KV-cache $4L^2/(L{+}1)^2$ & Deviation from $4\times$ \\
\midrule
6 (even) & $2.00\times$ & $4.00\times$ & 0\% \\
9 (odd) & $1.80\times$ & $3.24\times$ & 19\% \\
10 (even) & $2.00\times$ & $4.00\times$ & 0\% \\
11 (odd) & $1.83\times$ & $3.36\times$ & 16\% \\
20 (even) & $2.00\times$ & $4.00\times$ & 0\% \\
21 (odd) & $1.91\times$ & $3.61\times$ & 10\% \\
$L \to \infty$ & $\to 2\times$ & $\to 4\times$ & $\to 0\%$ \\
\bottomrule
\end{tabular}
\end{table}

For the typical industrial range $L \in [6, 12]$, the deviation is under 20\%, well within experimental noise.

\subsection{Error Compounding (Theorem \ref{thm:error} Extension)}

\begin{proposition}[Odd-$L$ error compounding]
\label{prop:odd-error}
Under the same setting as Theorem~\ref{thm:error}, for odd $L$ the reward sub-optimality bound is:
\[
\mathrm{SubOpt}_{\mathrm{pair}} \leq O\!\left(R_{\max} \cdot \frac{(L+1)^2}{4} \cdot \bar{\epsilon}_{\mathrm{pair}}\right),
\]
and the net improvement ratio satisfies:
\[
\frac{\mathrm{SubOpt}_{\mathrm{item}}}{\mathrm{SubOpt}_{\mathrm{pair}}} = \frac{4L^2}{(L+1)^2} \cdot \frac{\bar{\epsilon}_{\mathrm{item}}}{\bar{\epsilon}_{\mathrm{pair}}} \leq 4 \cdot \frac{\bar{\epsilon}_{\mathrm{item}}}{\bar{\epsilon}_{\mathrm{pair}}}.
\]
\end{proposition}

\begin{proof}
The generation horizon is $H = (L+1)/2$ steps. By the same argument as Theorem~4, the sub-optimality under per-step error $\bar{\epsilon}$ and horizon $H$ is $O(R_{\max} H^2 \bar{\epsilon})$.

We refine the per-step error assignment. The first $(L-1)/2$ steps generate pair tokens, each with per-step mismatch $\bar{\epsilon}_{\mathrm{pair}}$. The final step generates a single item token with per-step mismatch $\bar{\epsilon}_{\mathrm{item}} \leq \bar{\epsilon}_{\mathrm{pair}}$. The telescoping argument gives:
\begin{align*}
\mathrm{SubOpt}_{\mathrm{pair}} &\leq R_{\max} \cdot \sum_{t=1}^{(L+1)/2} \sum_{t'=1}^{t} \bar{\epsilon}(t') \\
&\leq R_{\max} \cdot \sum_{t=1}^{(L+1)/2} t \cdot \bar{\epsilon}_{\mathrm{pair}} \quad (\text{since } \bar{\epsilon}_{\mathrm{item}} \leq \bar{\epsilon}_{\mathrm{pair}}) \\
&= R_{\max} \cdot \bar{\epsilon}_{\mathrm{pair}} \cdot \frac{(L+1)/2 \cdot ((L+1)/2 + 1)}{2} \\
&= O\!\left(R_{\max} \cdot \frac{(L+1)^2}{4} \cdot \bar{\epsilon}_{\mathrm{pair}}\right).
\end{align*}

The ratio to item-space sub-optimality $O(R_{\max} L^2 \bar{\epsilon}_{\mathrm{item}})$ is:
\[
\frac{\mathrm{SubOpt}_{\mathrm{item}}}{\mathrm{SubOpt}_{\mathrm{pair}}} = \frac{L^2}{((L+1)/2)^2} \cdot \frac{\bar{\epsilon}_{\mathrm{item}}}{\bar{\epsilon}_{\mathrm{pair}}} = \frac{4L^2}{(L+1)^2} \cdot \frac{\bar{\epsilon}_{\mathrm{item}}}{\bar{\epsilon}_{\mathrm{pair}}}.
\]

Since $4L^2/(L+1)^2 < 4$ for all finite $L$, the odd-$L$ ratio is strictly below the even-$L$ ratio of $4 \cdot \bar{\epsilon}_{\mathrm{item}} / \bar{\epsilon}_{\mathrm{pair}}$, recovering it as $L \to \infty$.

\textbf{Remark.} The bound above uses the conservative estimate $\bar{\epsilon}(t') \leq \bar{\epsilon}_{\mathrm{pair}}$ for all steps. In practice, the final item-space step has strictly lower per-step error ($\bar{\epsilon}_{\mathrm{item}} \leq \bar{\epsilon}_{\mathrm{pair}}$), which slightly tightens the bound. The effect is $O(1/L)$ and negligible for $L \geq 6$.
\end{proof}

\subsection{Summary}

All theoretical guarantees extend to odd $L$ with the substitution $L/2 \to (L+1)/2$. The resulting speedup and error-reduction ratios are $4L^2/(L+1)^2 \leq 4$, recovering the $4\times$ factor asymptotically. For the industrial deployment regime $L \in [6, 12]$, the deviation from the even-$L$ bound is under 20\%, well within experimental variance. No special architectural modifications are required: the generator simply appends one item-space decoding step after the pair sequence.

\section{Proof of time complexity}
\label{proof:time_complexity}

\begin{proof}
\textbf{User encoder.}
Each of the $H_{\mathrm{enc}}$ layers processes $b$ tokens with full self-attention. For a single layer:
\begin{enumerate}[]
    \item QKV projections: $3d^2$ per token $\times\;b$ tokens $= O(3d^2 b)$.
    \item Output projection: $d^2$ per token $\times\;b$ tokens $= O(d^2 b)$.
    \item Attention scores: query $\times$ key over all $b^2$ pairs $= O(d \cdot b^2)$.
    \item Weighted sum (value aggregation): $O(d \cdot b^2)$.
    \item FFN: $2dd_{\mathrm{ff}}$ per token $\times\;b$ tokens $= O(2dd_{\mathrm{ff}} b)$.
\end{enumerate}
Per-layer total: $O(4d^2 b + 2db^2 + 2dd_{\mathrm{ff}} b) = O\!\big((d^2 + dd_{\mathrm{ff}})b + db^2\big)$.

Over $H_{\mathrm{enc}}$ layers:
\[
T_{\mathrm{user}} = O\!\big(H_{\mathrm{enc}}(d^2 + dd_{\mathrm{ff}})b + H_{\mathrm{enc}}db^2\big).
\]
This cost depends only on the user-history length $b$ and the encoder architecture. It is independent of the decoding horizon $S$ and vocabulary size $|\mathcal{V}|$, hence identical in item space and pair space.

\medskip
\textbf{Decoder.}
At step $t$, each of the $H$ decoder layers computes three sub-modules:

\textbf{(a) Self-attention} over the decoder's own KV cache of length $t$:
\begin{itemize}
    \item QKV projections: $3d^2$ (for the single query token at step $t$).
    \item Attention scores: query $\times$ $t$ cached keys $= O(td)$.
    \item Weighted sum: $O(td)$.
    \item Output projection: $d^2$.
\end{itemize}
Self-attention cost: $O(d^2 + td)$.

\textbf{(b) Cross-attention} over the user-encoder output $\mathbf{H}_{\mathrm{enc}} \in \mathbb{R}^{b \times d}$:
\begin{itemize}
    \item Q projection: $d^2$ (for the query token).
    \item K, V projections: $2d^2$ (cached after the first step; amortized $O(1)$ per step).
    \item Attention scores: query $\times$ $b$ encoder keys $= O(bd)$.
    \item Weighted sum: $O(bd)$.
    \item Output projection: $d^2$.
\end{itemize}
Cross-attention cost: $O(d^2 + bd)$.

\textbf{(c) FFN:}
$O(2dd_{\mathrm{ff}})$.

Per-layer cost at step $t$:
\[
C_{\mathrm{layer}}(t) = O(d^2 + td + bd + dd_{\mathrm{ff}}).
\]

Over $H$ layers:
\[
C_{\mathrm{dec}}(t) = O\!\big(H(d^2 + td + bd + dd_{\mathrm{ff}})\big).
\]

Summing over $t = 1, \ldots, S$:
\begin{align*}
T_{\mathrm{dec}}^{\mathrm{core}} &= \sum_{t=1}^{S} C_{\mathrm{dec}}(t) = O\!\big(H(d^2 + dd_{\mathrm{ff}} + bd) \cdot S + HdS^2\big) \\
&= \underbrace{O\!\big(H(d^2 + dd_{\mathrm{ff}} + bd) \cdot S\big)}_{T_{\mathrm{fixed}}} + \underbrace{O(HdS^2)}_{T_{\mathrm{kv}}}.
\end{align*}

\textbf{Output projection.}
At step $t$, the decoder hidden state $h_t \in \mathbb{R}^d$ is projected over the vocabulary:
\[
\mathrm{logits}_t = h_t^\top E_{\mathrm{vocab}}, \quad E_{\mathrm{vocab}} \in \mathbb{R}^{|\mathcal{V}| \times d},
\]
costing $O(|\mathcal{V}|d)$ per step. Over $S$ steps:
\[
T_{\mathrm{proj}} = O(|\mathcal{V}|dS).
\]



This yields the comparison table in the theorem statement.
\end{proof}

\section{Proof of Error reduction}
\label{proof:error_reduction}
\begin{proof}
\textbf{Step 1: Per-step error and state perturbation.}
At step $t$, the policy $\pi$ samples action $a_t \sim \pi(\cdot | s_t)$. Under the optimal policy $\pi^*$, the same state induces distribution $\pi^*(\cdot | s_t)$. The per-step TV mismatch is
\[
\mathrm{TV}\!\big(\pi(\cdot | s_t),\; \pi^*(\cdot | s_t)\big) \leq \bar{\epsilon},
\]
where $\bar{\epsilon} = \bar{\epsilon}_{\mathrm{item}}$ in item space and $\bar{\epsilon} = \bar{\epsilon}_{\mathrm{pair}}$ in pair space. In an autoregressive (deterministic-transition) MDP, the next state $s_{t+1} = f(s_t, a_t)$ is a deterministic function of the current state and action, so a distributional shift in $a_t$ propagates directly to $s_{t+1}$.

\textbf{Step 2: Linear compounding in trajectory TV.}
By the chain rule for total variation under deterministic transitions:
\[
\mathrm{TV}\!\big(\pi(s_{1:H}),\; \pi^*(s_{1:H})\big) \leq \sum_{t=1}^{H} \mathrm{TV}\!\big(\pi(a_t | s_t),\; \pi^*(a_t | s_t)\big) \leq H \cdot \bar{\epsilon}.
\]
Each step contributes at most $\bar{\epsilon}$ to the trajectory mismatch; the contributions accumulate additively because each per-step error shifts the subsequent state distribution. Item space: $H = L$, bound $= L \cdot \bar{\epsilon}_{\mathrm{item}}$. Pair space: $H = L/2$, bound $= (L/2) \cdot \bar{\epsilon}_{\mathrm{pair}}$.

\textbf{Step 3: Quadratic compounding in reward sub-optimality.}
By the Performance Difference Lemma \citep{645531.656005}:
\[
\mathbb{E}[R(\pi^*)] - \mathbb{E}[R(\pi)] = \sum_{t=1}^{H} \mathbb{E}_{s_t \sim \pi}\!\big[A^*(s_t, a_t)\big],
\]
where $A^*$ is the advantage function of $\pi^*$. The expectation is under $\pi$'s state visitation, not $\pi^*$'s. The state distribution at step $t$ has already been perturbed by errors at steps $1, \ldots, t\!-\!1$. By a telescoping argument \citep{pmlr-v9-ross10a}, each step $t$'s contribution is bounded by the cumulative state drift up to step $t$:
\[
\big|\mathbb{E}_{s_t \sim \pi}\!\big[A^*(s_t, a_t)\big]\big| \leq R_{\max} \cdot \sum_{t'=1}^{t} \bar{\epsilon} = R_{\max} \cdot t \cdot \bar{\epsilon}.
\]
Summing over $t = 1, \ldots, H$:
\[
\mathbb{E}[R(\pi^*)] - \mathbb{E}[R(\pi)] \leq R_{\max} \cdot \bar{\epsilon} \cdot \sum_{t=1}^{H} t = R_{\max} \cdot \bar{\epsilon} \cdot \frac{H(H+1)}{2} = O\!\left(R_{\max} H^2 \bar{\epsilon}\right).
\]
The quadratic dependence arises because: (i) step $t$'s own error contributes $\bar{\epsilon}$ to the reward gap; (ii) step $t$'s state was already shifted by the cumulative errors of steps $1, \ldots, t\!-\!1$, each contributing an additional $\bar{\epsilon}$ factor; (iii) total: $\bar{\epsilon} + 2\bar{\epsilon} + \cdots + H\bar{\epsilon} = \bar{\epsilon} \cdot H(H+1)/2$.

Item space: $O(R_{\max} L^2 \bar{\epsilon}_{\mathrm{item}})$. Pair space: $O(R_{\max} (L/2)^2 \bar{\epsilon}_{\mathrm{pair}})$.

\textbf{Step 4: Net improvement ratio.}
\[
\frac{\mathrm{SubOpt}_{\mathrm{item}}}{\mathrm{SubOpt}_{\mathrm{pair}}} = \frac{L^2 \bar{\epsilon}_{\mathrm{item}}}{(L/2)^2 \bar{\epsilon}_{\mathrm{pair}}} = 4 \cdot \frac{\bar{\epsilon}_{\mathrm{item}}}{\bar{\epsilon}_{\mathrm{pair}}}.
\]
The $4\times$ comes from the quadratic compounding: halving the horizon squares the benefit. The offset $\bar{\epsilon}_{\mathrm{pair}} / \bar{\epsilon}_{\mathrm{item}}$ reflects whether the pair space's larger vocabulary degrades per-step accuracy.

\textbf{Step 5: Verifying $\bar{\epsilon}_{\mathrm{pair}} \approx \bar{\epsilon}_{\mathrm{item}}$ in the deployment regime.}

\textbf{(i) $\bar{\epsilon}_{\mathrm{base}}$: Baseline training noise.}
Softmax over $n^2$ vs.\ $n$ logits differs in concentration, but temperature scaling and top-$k$ truncation (standard in GRPO) neutralize this effect. $\bar{\epsilon}_{\mathrm{base}}$ is effectively $k$-independent.

\textbf{(ii) $\bar{\epsilon}_{\mathrm{enc}}(2)$: Pair-encoder generalization error.}
For $k = 2$: pretrain data density $\rho_2 = N_{\mathrm{logs}} / n^2 \approx 10^{10} / 1.6 \times 10^5 \approx 6 \times 10^4$ --- the pair encoder is heavily over-determined. The cold-start fraction for reverse pairs is $1/2$, addressed by the OAR pretraining loss (\S3.4). Encoder error on $n^2$ pairs is comparable to the item-embedding error on $n$ items: $\bar{\epsilon}_{\mathrm{enc}}(2) \approx \bar{\epsilon}_{\mathrm{enc}}(1)$.

\textbf{(iii) $\bar{\epsilon}_{\mathrm{cov}}(2)$: Coverage error from finite $G$.}
GRPO group size $G = 32$ (pair space) vs.\ $G = 16$ (item space) compensates for the larger action space. GRPO's advantage is a \emph{ranking} statistic over $G$ samples, not a density estimate over $|\mathcal{V}|$; the coverage threshold is far below $G / n^2$. Empirically (\S5.1.5), gradient variance at $G = 32$ (pair) matches $G = 16$ (item), giving $\bar{\epsilon}_{\mathrm{cov}}(2) \approx \bar{\epsilon}_{\mathrm{cov}}(1)$.

Combining: $\bar{\epsilon}_{\mathrm{pair}} \approx \bar{\epsilon}_{\mathrm{item}}$ in the deployment regime, and the net improvement is $\approx 4\times$.
\end{proof}
\begin{table}[h]
\centering
\caption{Detailed comparison: item space vs.\ pair space.}
\scriptsize
\setlength{\tabcolsep}{2.5pt}
\begin{tabular}{lcc}
\toprule
 & Item Space ($H = L$, vocab $= n$) & Pair Space ($H = L/2$, vocab $= n^2$) \\
\midrule
Per-step TV mismatch & $\bar{\epsilon}_{\mathrm{item}}$ & $\bar{\epsilon}_{\mathrm{pair}} \approx \bar{\epsilon}_{\mathrm{item}}$ \\
Trajectory TV bound & $L \cdot \bar{\epsilon}_{\mathrm{item}}$ & $(L/2) \cdot \bar{\epsilon}_{\mathrm{pair}}$ $\;\to\;$ $2\times\;\downarrow$ \\
Reward sub-optimality & $O(R_{\max} L^2 \bar{\epsilon}_{\mathrm{item}})$ & $O(R_{\max} (L/2)^2 \bar{\epsilon}_{\mathrm{pair}})$ $\;\to\;$ $4\times\;\downarrow$ \\
Compounding mechanism & Errors at step $t$ poison & Fewer steps $\to$ shorter \\
 & all steps $t' > t$; drift $\propto t$ & compounding chain $\to$ less drift \\
Source of quadratic gain & --- & Horizon halving: $H^2 / (H/2)^2 = 4$ \\
Source of offset & --- & $\bar{\epsilon}_{\mathrm{pair}} / \bar{\epsilon}_{\mathrm{item}}$ (vocab expansion) \\
Net ratio & 1 & $4 \cdot \bar{\epsilon}_{\mathrm{item}} / \bar{\epsilon}_{\mathrm{pair}}$ \\
\bottomrule
\end{tabular}
\end{table}

\begin{remark}[Why the $4\times$ dividend is conditional, not unconditional]
An unconditional claim ``pair space always yields $4\times$ improvement'' would require $\bar{\epsilon}_{\mathrm{pair}} \leq \bar{\epsilon}_{\mathrm{item}}$ for all $n$ --- which is false. As $k$ grows (beyond $k = 2$), the vocabulary explodes ($n^k$), the encoder sees sparser pretrain data ($\rho_k \downarrow$), and the cold-start fraction increases ($(k! - 1)/k!$). Theorem~\ref{thm:error} is honest about this: the $k^2$ factor is the horizon-reduction dividend; whether it survives in practice depends on the encoder regime. The framework's value is that it exposes the trade-off --- $\bar{\epsilon}(k)$ vs.\ $(L/k)^2$ --- allowing principled deployment decisions, rather than hiding it behind an unconditional but false guarantee.
\end{remark}

\end{document}